\documentclass[11pt,a4paper]{article}
\usepackage{fullpage,xspace,times}
\usepackage{graphicx}
\usepackage{amsmath,upgreek}
\usepackage{amsfonts}
\usepackage{amssymb}
\usepackage{epstopdf}
\usepackage{xcolor}

\newtheorem{theorem}{Theorem}

\newtheorem{claim}[theorem]{Claim}

\newtheorem{corollary}[theorem]{Corollary}

\newtheorem{definition}[theorem]{Definition}

\newtheorem{lemma}[theorem]{Lemma}

\newtheorem{proposition}[theorem]{Proposition}

\newcommand{\OPT}{\mbox{\textsc{opt}}}
\newcommand{\opt}{\mbox{\textsc{opt}}}
\newcommand{\sol}{\mbox{\textsc{sol}}}
\newcommand{\optn}{\mbox{\textsc{optn}}}

\newcommand{\pr}{\mbox{\textsc{ooebp}}}

\newcommand{\C}{{\cal{C}}}

\newcommand{\Xomit}[1]{ }
\newenvironment{proof}[1][Proof]{\textbf{#1.} }{\ \rule{0.5em}{0.5em}}

\mathchardef\mhyphen="2D

\newcommand{\eps}{\upvarepsilon}

\begin{document}

\title{More on ordered open end bin packing}

\date{}

\author{J\'anos Balogh \thanks{Institute of Informatics,
     University of Szeged, Szeged, Hungary. \texttt{baloghj@inf.u-szeged.hu}. Supported by the project ``Integrated program for
training new generation of scientists in the fields of computer
science'', no. EFOP-3.6.3-VEKOP-16-2017-0002.}  \and Leah
Epstein\thanks{ Department of Mathematics, University of Haifa,
Haifa, Israel. \texttt{lea@math.haifa.ac.il}. } \and Asaf
Levin\thanks{Faculty of Industrial Engineering and Management, The
Technion, Haifa, Israel. \texttt{levinas@ie.technion.ac.il.}
Partially supported by grant number 308/18 of ISF - Israeli
Science Foundation.}}

\maketitle

\begin{abstract}
We consider the Ordered Open End Bin Packing problem. Items of
sizes in $(0,1]$ are presented one by one, to be assigned to bins
in this order. An item can be assigned to any bin for which the
current total size strictly below $1$. This means also that the
bin can be overloaded by its last packed item. We improve lower
and upper bounds on the asymptotic competitive ratio in the online
case. Specifically, we design the first algorithm whose asymptotic
competitive ratio is strictly below $2$ and it is close to the
lower bound. This is in contrast to the best possible absolute
approximation ratio, which is equal to $2$. We also study the
offline problem where the sequence of items is known in advance,
while items are still assigned to bins based on their order in the
sequence. For this scenario we design an asymptotic
polynomial time approximation scheme.
\end{abstract}

\section{Introduction}
We study Ordered Open End Bin Packing (\pr). The input for this
problem is a sequence of items of positive sizes. An item can be
assigned to any bin that has a current total size strictly smaller
than $1$. In the online problem, items are presented one by one to
be packed in this way. In contrast, an offline algorithm knows the
sequence of  input items in advance. Since the input is ordered,
it also has to process the input as a sequence when it creates a
packing, where the input is ordered in the same way as it would
have been presented to an online algorithm.

We analyze algorithms via worst-case analysis. The absolute
competitive ratio (for online algorithms) or absolute
approximation ratio (for offline algorithms) is the worst-case
ratio between the cost of the algorithm and the optimal (offline)
cost (for the same input). The asymptotic measures are the
superior limits of these values when we let the optimal cost grow
to infinity. The asymptotic measures are known to be the
meaningful ones for bin packing problems, and thus, in this paper,
we will sometimes omit the word asymptotic. An optimal offline
solution is denoted by $\OPT$, and its cost for an input $I$ is
denoted by $\OPT(I)$.

There are several variants for open end bin packing. The total
size of items packed into a bin is called {\it load}, and in all
these variants it is possible to pack sets of items into bins with
loads above $1$ under certain conditions. In the maximum variant,
it is required that every bin has some ordering of items such that
the removal of the last item results in load strictly below $1$.
Thus, in this version, it is required that the load will be below
$1$ after the removal of the largest item. In the minimum variant,
it is required that for every bin, the removal of any item causes
the load to be below $1$. Thus, in this version, the condition is
on the removal of the smallest item. These two variants and ones
that are equivalent to them were studied under different names
\cite{Zhang98,LDY01,EL08,LYX10A,LYX10T,GZ09}. These variants are
generally very different from \pr. In the maximum variant, an
asymptotic fully polynomial time approximation scheme (AFPTAS) was
obtained by applying the AFPTAS for standard bin packing
\cite{FerLue81,KK82,BJK18} and packing the largest items as the
last items of their bins \cite{LDY01}. Such a scheme is a family
of algorithms, where for every $\eps>0$ there is an algorithm of
asymptotic approximation ratio at most $1+\eps$ and the running
time is polynomial in the input size and in $\frac 1{\eps}$. In
the minimum variant, an AFPTAS was designed as well
\cite{LYX10T,EL08}, which required additional ideas, but it cannot
be adapted for \pr\ due to the ordered input. More precisely, in
the ordered variant, in some cases the last item of a bin may be
the item of maximum size, in other cases it can only be the item
of minimum size, and typically, the last item (which is just the
item of the largest index) is not the maximum or the minimum.
Thus, offline algorithms need to be designed carefully such that
there is full knowledge on the identity of the last item. Another
variant for online algorithms is where the last item is the one
whose removal should bring the load below $1$, where the online
algorithm is compared to an offline algorithm that can reorder the
items. We will refer to this variant as the unfair variant.

The problem which we study, \pr, was studied Yang and Leung
\cite{YangL03}. For this problem, we will use the term $1$-items
for items of size $1$ or larger. As explained in \cite{YangL03},
such an item has a special role since once it is packed into a
bin, even if the bin is empty, no additional items can be packed
into this bin. We assume that all these items have size $1$
exactly, since all items of sizes at least $1$ are equivalent with
respect to the action of any algorithm. Thus, input items have
rational sizes in $(0,1]$. In addition to average-case analysis,
Yang and Leung \cite{YangL03} design an algorithm whose asymptotic
competitive ratio is strictly below $2$ for the case without
$1$-items, and they show that the algorithm cannot perform much
better if there are $1$-items. They also proved lower bounds on
the asymptotic competitive ratio for the online case: $1.630297 $
for inputs with $1$-items, and $1.415715$ for inputs without such
items. The lower bound results are proved based on a computer
assisted proof, where all packing patterns are enumerated. While
the reader has no access to this analysis, we have verified these
results independently using a different method (see discussion
below on our results). One can observe that upper bounds for the
unfair variant are valid for \pr, since the actions of the online
algorithm and its objective function value are the same for both
variants, while an offline algorithm for the unfair variant can
perform all actions which it can do for \pr, and possibly other
actions, so it may have a smaller cost for the unfair variant but
not a larger cost. For this variant, tight asymptotic competitive ratio bounds of
$2$ and $1.5$ are known for the cases with and without $1$-items
respectively \cite{LDY01,Zhang98} (the bound of $2$ is tight also
for the absolute competitive ratio).

We consider both the offline variant of the problem for which we
design an asymptotic approximation scheme and the online case for
which we design the current best online algorithms and improved
lower bounds on the possible asymptotic competitive ratio that can
be achieved by online algorithms.

Note that \pr\ is different from classic online bin packing, for
which the current best lower and upper bounds on the asymptotic
competitive ratio are $1.542780906$ and $1.57828956$, respectively
\cite{BBDEL_newlb,BBDEL_ESA18}. Recall that the absolute
competitive ratio is usually seen as a less  interesting measure
for bin packing problems. It is known that its value for classic
online bin packing is $\frac 53$ \cite{Zhang,BBDSS19}. There are
several other packing problems where an offline solution still
needs to process the input as a sequence
\cite{FM,E_LIB,dosa2013bin,CSW11,BBDEKLT15,BBDEKT15,BDESV18}.

\paragraph{Our results.}

We design an asymptotic polynomial time approximation scheme
(APTAS) for the offline variant. Such a scheme still has an
algorithm of asymptotic approximation ratio at most $1+\eps$ for
every $\eps>0$, but the running time is not necessarily polynomial
in $\frac 1{\eps}$, that is, $\eps$ is seen as a constant and for every fixed value of $\eps$ the time complexity is polynomial. As
explained above, the model for \pr\ is very different from other
variants of open end bin packing, and other bin packing problems.
The authors are not aware of any asymptotic approximation schemes for packing
problems over sequences, and previously known approximation
schemes are for the variants where offline solutions can reorder
the input \cite{BBDEKLT15,LDY01,LYX10T,EL08}.

For the online problem, we briefly discuss the relation between
variants, and show that the absolute competitive ratio for \pr\ is
exactly $2$. Then, we analyze the asymptotic competitive ratio
using a combination of new and old methods. We define a new class
of algorithms, which allows us to improve the upper bound on the
asymptotic competitive ratio from $2$ \cite{LDY01} to
approximately $1.691561$. We show that the obtained ratio is tight
for the class of algorithms which we define. We design a similar
algorithm for the case without $1$-items, which yields an
asymptotic competitive ratio of at most $1.44465$, improving over
the previous bound of $1.5$ \cite{Zhang98}. As mentioned above,
lower bounds on the asymptotic competitive ratio were given with
partial proofs \cite{YangL03}. We fill this gap and show that they
can be improved slightly using a different input \cite{BBG}. Thus,
the gaps for the asymptotic competitive ratios are now between
approximately $1.630483$ and approximately $1.691561$ for the case
with $1$-items, and between approximately $1.415752$ and $1.44465$
for the case without $1$-items.

\section{An Asymptotic polynomial time approximation scheme (APTAS) for the (offline) \pr}
Let $\eps>0$ be such that $\frac 1{\eps} \geq 3$ is an integer
(and in particular, $\eps \leq \frac 13$). When we will consider
multiple instances of \pr, it will be useful to denote by
$\opt(I)$ the optimal cost for instance $I$, but when the instance
is clear by context we use $\opt$ to denote this optimal cost.

In order to design an asymptotic polynomial time approximation
scheme (APTAS), it suffices to show the existence of a polynomial
time algorithm that always returns a solution of cost at most
$(1+\eps)^c\opt +f(1/\eps)$ for a positive constant $c \geq 1$ and
some function $f$ (where $\opt$ is the optimal cost for the same
instance). The degree of the polynomial upper bounding the time
complexity of this algorithm may depend on $\eps$.

Our scheme applies a guessing step, where this step is followed by
a pre-processing step that applies a linear grouping type of
rounding \cite{FerLue81,KK82}. Then, it uses an algorithm for
solving a fixed-dimension integer program (IP) based on a
configuration IP, in order to create a plan of the output
\cite{Len83,Kan83}. This last plan is transformed into a feasible
solution for \pr\ in the final post-processing step.  The guessing
step (together with a modification of the solution based on it) is
the novel step which allows us to adjust the known methods, which
were previously used for problems without sequence-dependent
information, for our problem. This step allows us to overcome the
complications of designing algorithms for inputs that are
sequences rather than sets.

\subsection{The guessing step}
Let $1,2,\ldots ,n$ be the sequence of items that is given as
input, and let $s_i$ be the size of item $i$.  We define the
exceeding item of a bin in the following way.
\begin{definition}
Fix a bin $B$ in a feasible solution. The {\em exceeding item} of
$B$, if the total size of items in $B$ is at least $1$, is the
item of maximum index packed into $B$.  For a bin $B$ with total
size of items strictly smaller than $1$, its exceeding item is
undefined (and it has no exceeding item).
\end{definition}

We note that based on this definition, a 1-item is always an
exceeding item of a bin, even if it is the unique item of a bin.

Let $\opt$ be a fixed optimal solution.  We next show that we can
assume that an exceeding item has size of at least $\eps^2$, while
we still get a maintainable near-optimal solution, denoted by
$\OPT'$.
\begin{lemma}
Given an optimal solution $\OPT$ of cost $\opt$, there is a feasible solution $\OPT'$ of cost at most $(1+\eps^2)\opt+1$ such that every exceeding item has size at least $\eps^2$.
\end{lemma}
\begin{proof}
Consider the set of exceeding items of sizes smaller than $\eps^2$
in $\OPT$.  We repack these items into new bins such that there
are $\frac{1}{\eps^2}$ such items packed into each bin, except
perhaps for the last bin that may have a smaller number of items.
Let $\OPT'$ denote the resulting solution.  Observe that by
definition, in $\OPT'$ all exceeding items are of size at least
$\eps^2$. This holds as all the smaller items that were exceeding
items are repacked into bins where no such bin has an exceeding
item. Furthermore, the number of new bins is at most
$\eps^2\opt+1$, and thus the claim follows.
\end{proof}

We will establish the existence of a near optimal solution that
has a {\em certificate} which we define as follows.
\begin{definition}
Let $\sol$ be a feasible solution for \pr.  We say that $\sol$ has
a certificate $$(e_0,e_1,e_2,\ldots ,e_{1/\eps})$$ if the
following conditions hold:
\begin{enumerate}
\item $0 = e_0 < e_1 \leq e_2 \leq \cdots \leq e_{1/\eps}=n$ and
$e_1,e_2,\ldots ,e_{1/\eps-1}$ are integers. \item For every bin
$B$ in $\sol$ (exactly) one of the following cases holds:
\begin{itemize}
\item Either the total size of the items in $B$ is strictly
smaller than $1$, i.e., $B$ does not have an exceeding item, \item
or there is an integer $i(exceed)$ such that the exceeding item
has an index strictly larger than $e_{i(exceed)}$ and all other
items in $B$ (if there are any such items) have indices at most
$e_{i(exceed)}$.
\end{itemize}
\end{enumerate}
If this holds for every bin, we say that $(e_0,e_1,e_2,\ldots
,e_{1/\eps})$ is a {\em certificate} of $\sol$.
\end{definition}

Every input has at least one solution with a certificate.
Specifically, the certificate $(0,n,\ldots,n)$ is a certificate of
a solution where every item is packed into a different bin. This
holds since the only bins with exceeding items are those with
$1$-items, and for each such item, as $e_0=0$, its index is
larger.

\begin{definition}
A solution $\sol$ for \pr\ is called a {\em nice solution} if it
satisfies the following conditions.  First, for every bin $B$ in
$\sol$ that has an exceeding item, the  size of the exceeding item
of $B$ is at least $\eps^2$, second, $\sol$ has a certain
certificate $(e_0,e_1,e_2,\ldots ,e_{1/\eps})$.
\end{definition}

Let $\optn$ be an optimal nice solution (i.e., a solution of
minimal cost among the nice solutions). Next, we show that we can
approximate $\optn$.  With a slight abuse of notation we denote by
$\optn$ both the solution and its cost, and we note that in every
case the distinction between the two will be clear by context.
Similarly, we will use $\opt'$ to denote the cost of $\OPT'$.

The next lemma also shows in particular that there is at least one
nice solution for every input.

\begin{lemma}\label{nice_lem}
We have $\optn \leq (1+\eps)\cdot \opt' +\frac{1}{\eps}$.
\end{lemma}
\begin{proof}
Recall that $\OPT'$ of cost $\opt'$ satisfies the first condition in the definition of nice solutions.  We create a nice solution $\sol$ by modifying $\OPT'$ such that the first condition will be maintained while the second condition will be satisfied as well.  Then we will show that the cost of $\sol$ is at most $ (1+\eps)\cdot \opt'+\frac{1}{\eps}$.  The claim will follow by the optimality of $\optn$ among nice solutions.

Consider the solution $\OPT'$.  Some bins in this solution have
exceeding items while other bins do not have exceeding items.  The
packing of items that were packed into bins (of $\OPT'$) without
exceeding items is left without modification.  Consider the $m$
bins $O_1,O_2,\ldots ,O_m$ of $\OPT'$ containing exceeding items,
such that these bins are sorted according to an increasing order
of the indices of the exceeding items of these bins.  Let the
integer $\alpha=\lfloor \eps m \rfloor$ be the result of integer
division of $m$ by $\frac 1{\eps}$. The remainder of this division
is $m-\frac{\alpha}{\eps}$. Let $\beta = m-(\frac 1{\eps}-1)\cdot
\alpha$ be the sum of $\alpha$ and the above remainder.

%If $m\leq \frac{1}{\eps}$ we modify $\OPT'$ by repacking the
%exceeding items of $\OPT'$ in new bins one item per bin, and let
%$\sol$ be the new solution.  We define a certificate vector
%$(e_0,e_1,e_2,\ldots ,e_{1/\eps})$ by letting $e_i=0$ for all
%$i<\frac 1{\eps}$ (and $e_{1/\eps}=n$ as required). Observe that
%this is a certificate of $\sol$ as every bin in $\sol$ either does
%not have an exceeding item or contains an exceeding item but no
%other item (this can happen for items of size $1$, their bins will
%not have other items, so the condition will hold for $e_1$). In
%this case $\sol$ is indeed nice and its cost is at most
%$\opt'+\frac{1}{\eps}$
%

In the case where $m \geq \frac{1}{\eps}$, it holds that $\alpha
\geq 1$ and otherwise $\alpha=0$. In both cases we have $\beta
\leq \eps\cdot m +\frac{1}{\eps}$ and $\beta \geq \eps \cdot m
\geq \alpha$.

In order to modify $\OPT'$ we do the following. The exceeding
items of $O_1,\ldots, O_{\beta}$ are packed into new bins, one
item per bin. If the unique item packed into the bin is a 1-item,
then it is of size at least $\eps^2$, and it is an exceeding item
but no other item is packed there so both conditions hold by using
$i(exceed)=0$ for these bins. Otherwise, the bin does not contain
an exceeding item so both conditions hold trivially. Thus, the
bins obtained in this way satisfy both conditions no matter which
certificate vector we consider. In the case $m < \frac 1{\eps}$ we
are done as $\alpha = 0$ and $\beta = m$ hold in this case. In
particular, in the last case $\sol$ is indeed nice and its cost is
at most $\opt'+\frac{1}{\eps}$.

Otherwise, $\beta<m$, and we deal with the remaining $m-\beta$
bins with exceeding items. Prior to the last step all bins
$O_1,\ldots, O_{\beta}$ had exceeding items as well, where after
this step only bins $O_{\beta+1},O_{\beta+2},\ldots, O_{m}$ have
exceeding items.

For every bin of $O_p$ such that $\beta+1 \leq p \leq m$, we
repack the exceeding item of $O_{p}$ into the bin $O_{p-\alpha}$,
where we apply this for every such $p$. The smallest index of any
bin receiving an exceeding item is $\beta+1-\alpha \geq 1$.
Observe that the index of the new item joining a bin is larger
than the index of its original exceeding item due to the sorting
of these bins, so the packing remains valid. Thus, every repacked
item is the new exceeding item of its new bin, if this new bin
indeed has an exceeding item after the transformation.

We denote the resulting solution by $\sol$. Thus, every bin $B$ in
$\sol$ satisfies that if $B$ has an exceeding item, then the size
of the exceeding item of $B$ is at least $\eps^2$ because every
exceeding item of a bin in $\sol$ was an exceeding item of a bin
in $\OPT'$. This applies for the bins of $\OPT'$ that did not have
exceeding items as the packing of items into these bins is the
same as well. Note that after the transformation, bins
$O_{m-\alpha+1}, O_{m-\alpha+2}, \ldots, O_{m}$ have no exceeding
items.

We define the vector  $(e_0,e_1,e_2,\ldots ,e_{1/\eps})$ by
letting $\eps_0=0$ and $e_{1/\eps}=n$ as required, and for every
$1 \leq i \le \frac{1}{\eps}-1$ letting $e_i$ be the index of the
exceeding item of $O_{\beta+(i-1)\alpha}$ in the solution $\OPT'$
(i.e., the solution before the transformation). Observe that this
vector has a monotonically increasing list of components due to
the sorting of bins with exceeding items in $\OPT'$.

It suffices to show that for every bin $B$ in $\sol$, containing
both non-exceeding items as well as an exceeding item, we have a
value of $i$, such that all non-exceeding items have indices at
most $e_i$ while the exceeding item has index strictly larger than
$e_i$.

We next argue that this vector is a certificate of $\sol$. We only
consider the bins out of
$$O_{\beta+1-\alpha},O_{\beta+2-\alpha},\ldots,O_{m-\alpha}$$ (where
$m-\alpha=\beta+(\frac{1}{\eps}-2)\alpha$) with exceeding items
after the transformation, as other bins with exceeding items were
discussed already (those are only bins of $\sol$ with $1$-items as
their only items). For a given bin $B$, let $i'$ be such that its
index $\ell$ as a bin $O_{\ell}$ is in
$(\beta+(i'-2)\alpha,\beta+(i'-1)\alpha]$ (where $i' \in
\{1,2,\ldots,\frac 1{\eps}-1\}$). The item of index $e_{i'}$ was
the exceeding item of $O_{\beta+(i'-1)\alpha}$ in $\OPT'$,
therefore the index of the original exceeding item of $B=O_{\ell}$
in $\OPT'$, was at most $e_{i'}$ due to the sorting, and all its
other items have smaller indices since the bin was valid. The new
exceeding item of $B$ was previously the exceeding item of a bin
$O_{\ell'}$ where $\ell'$ is in
$(\beta+(i'-1)\alpha,\beta+i'\alpha]$, so its index is strictly
above $e_{i'}$. Thus, indeed this vector is a certificate of
$\sol$.

%Let $i$ be the maximum integer value such that the (new) exceeding
%item of $B$ has an index strictly larger than $e_i$.
%Observe that $e_i$ is well defined as the components of the
%certificate we consider are all distinct, and by our our sorting
%of the bins in $\OPT'$ with exceeding items we know that the
%exceeding item of $B$ was an exceeding item of a bin $O_q$ of
%$\OPT'$ with $q>\beta$.  By definition, the index of the exceeding
%item of $B$ is larger than $e_i$.   The non-exceeding items of $B$
%were packed into the bin $O_{q-\alpha}$ of $\OPT'$ and thus all of
%them have indices of at most the index of the item that used to be
%the exceeding item of $O_{q-\alpha}$ that is at most $e_i$ (by the
%sorting of the bins with exceeding items in $\OPT'$).

Furthermore, the process of transforming $\OPT'$ into the new solution, \sol, creates only $\beta \leq \eps \cdot m+ \frac{1}{\eps}$ new bins so the cost of $\sol$ is at most $(1+\eps) \opt' +\frac{1}{\eps}$ as we argued.
\end{proof}

\paragraph{The guessing.} We guess the certificate vector $(e_0,e_1,e_2,\ldots ,e_{1/\eps})$
of $\optn$ where the components of the vectors are integers in
$[1,n]$ ($\frac{1}{\eps} -1$ elements are guessed).  Thus, the number of
different values of the guesses is $O(n^{1/\eps})$. Each guess
will be examined in an iteration step of a loop in our (guessing)
procedure. For every value of the guess, we apply the algorithm in
the next step that returns a feasible solution for \pr, and among
all the solutions obtained in the different iterations of this
loop, we pick the cheapest one as the output of the algorithm.  In
order to analyze our algorithm it suffices to consider the
iteration of this loop in which we use the value of the guess
corresponding to a certificate of $\optn$, and show that for this
iteration the cost of the returned solution is at most
$(1+\eps)^c\optn +f(1/\eps)$ for a constant $c \geq 1$, and some
function $f$.

\subsection{The pre-processing step}
The pre-processing step which we apply is described as rounding of
large items, but we apply this operation separately for each
subsequence of items of indices in $(e_i,e_{i+1}]$ for every $i$.
Similarly to approximation schemes for the bin packing problem,
this rounding of large items is carried out using the so-called
{\it linear grouping} rounding method.  %%%%We let $e_0=0$.

In what follows we sometimes introduce (dummy) items in the middle
of the input sequence and sometimes we delete items from the
sequence.  In order to maintain the guessed certificate in these
operations, we treat the certificate as a collection of pointers
to a (doubly) linked list of the items in the input. Lists are
initialized by items of indices in $(e_i,e_{i+1}]$, where items
appear in the lists sorted by increasing indices.  Now, inserting
items means we insert items to the corresponding position in this
linked list, and deleting items is done as in linked lists (where
deletion can also be of the first or last item).
% that were pointed
%by one of the pointers (or more pointers) we mean that this
%pointer will point to the last item (appearing just before it in
%the linked list) that was not deleted prior to the previously
%pointed item.

In the remainder of this step we find an upper bound and a lower
bound on $\optn(I)$ of an instance $I$ of \pr\, by using two
different nice solutions (for the two bounds).

An item is called a {\it large item of interval $(e_i,e_{i+1}]$}
if its index is in the interval $(e_i,e_{i+1}]$ and its size is at
least $\eps^2$.  An item is called a {\it small item of interval
$(e_i,e_{i+1}]$} if its index is in the interval $(e_i,e_{i+1}]$
and its size is smaller than $\eps^2$.  An item is {\it large} if
it is a large item for some interval, and an item is an {\it item
of interval $(e_i,e_{i+1}]$} if its index is in the interval
$(e_i,e_{i+1}]$, that is, if it is either a large item for this
interval or a small item for this interval.  Throughout the rest
of the scheme, we keep the certificate vector as a requirement of
nice solutions in the sense that an exceeding item will belong to
a linked list with a larger value of $i$.
%%(this vector is modified if we introduce items or delete items according to the above rules).

We denote by $n_i$ the number of large items of interval
$(e_i,e_{i+1}]$ for input $I$.

\begin{lemma}\label{le} Without loss of generality, we assume that for every
$i$, we have that $\eps^3 n_i$ is an integer.
\end{lemma}
\begin{proof}
For values of $i$ for which the claim does not hold, we add up to
$\frac{1}{\eps^3}$ items, each of which has size $1$, and they
will appear in the sequence of items just before the item
$e_{i+1}$. Note that applying this transformation for all values
of $i$ that had not satisfied the claim may add up to
$\frac{1}{\eps^4}$ items so it increases the optimal cost of nice
solutions by an additive term of at most $\frac{1}{\eps^4}$. Thus,
it suffices to approximate the resulting instance after adding
these items.
\end{proof}

By slightly abusing notation, the input with the modification
described in the proof of Lemma \ref{le} is still denoted by $I$.

For every $i$, let $\ell^i(1),\ell^i(2),\ldots,\ell^i(n_i)$ be the
large items of interval $(e_i,e_{i+1}]$ sorted in a non-increasing
order of their size. That is,  $$s_{\ell^i(1)} \geq s_{\ell^i(2)}
\geq \cdots \geq s_{\ell^i(n_i)} \geq \eps^2 .$$  For every $i$,
and for every $k=1,2,\ldots ,\frac{1}{\eps^3}$, the $k$-th group
of interval $(e_i,e_{i+1}]$ denoted as $G(i,k)$ is the set of
$\eps^3 \cdot n_i$ items of indices $\ell^i((k-1)\cdot \eps^3
\cdot n_i+1),\ell^i((k-1)\cdot \eps^3 \cdot n_i+2), \ldots,
\ell^i(k\cdot \eps^3 \cdot n_i)$.

The {\it rounded instance} is the instance that we obtain by
rounding up the size of all large items such that for every $i,k$,
the items of $G(i,k)$ are rounded up to $s_{\ell^i((k-1)\cdot
\eps^3 \cdot n_i+1)}$ (that is the largest size of an item in
$G(i,k)$), while small items keep their original size. The small
items are not included in $G(i,k)$.

Recall that $I$ denotes the instance prior to this rounding, and
let $I'$ be the rounded instance.  Furthermore, we denote by $I''$
the instance obtained from $I'$ by deleting all items in
$\bigcup_{i=0}^{1/\eps-1} G(i,1)$, and observe that $I''$ can be
obtained from $I$ by rounding the size of each item of $G(i,k)$
down to the size of the largest item of $G(i,k+1)$ (for all $i$
and all $k<1/\eps^3$) and deleting all items of $G(i,1/\eps^3)$
(for all $i$). Note that we keep the small items separately, and
they are included in $I'$ and $I''$.

The use of this rounding is justified by the following lemma.

\begin{lemma}
We have $\optn(I'') \leq \optn(I) \leq \optn(I') \leq (1+3\eps)
\cdot \optn(I'')$.
\end{lemma}
\begin{proof}
The first two inequalities, i.e.,  $\optn(I'') \leq \optn(I) \leq
\optn(I')$ follow as when we decrease the size of some items and
perhaps delete some of those items a feasible nice solution with
respect to the certificate $(e_0,e_1,e_2,\ldots ,e_{1/\eps})$  for
the instance before the transformations remains feasible nice
solution with respect to the same certificate (some bins stop
having exceeding items but this does not hurt the property of
being nice). Thus the solution $\optn(I')$ is a feasible nice
solution for $I$ and so $ \optn(I) \leq \optn(I')$, and the
solution $\optn(I)$ is a feasible nice solution for $I''$ and so
$\optn(I'') \leq \optn(I)$.

It remains to prove the last inequality.  Given the solution
$\optn(I'')$ we create a solution for $I'$ by packing each item
that does not exist in $I''$ in its dedicated bin.  Note that the
resulting solution is obviously a feasible nice solution as
packing an item into a dedicated bin keeps it feasible and
maintain the property of being nice, and in total we added $\eps^3
\sum_{i=0}^{1/\eps-1} n_i$ bins.  Therefore,
\begin{equation}\optn(I') \leq \optn(I'')+ \eps^3
\sum_{i=0}^{1/\eps-1} n_i \label{eqq1} \ . \end{equation}

However, the instance $I''$ contains at least
$(1-\eps^3)\sum_{i=0}^{1/\eps-1} n_i$ large items, each of which
of size at least $\eps^2$.  Therefore, the optimal cost of a nice
solution (or any solution) is at least half the total size of
these items, since no bin can contain items of total size above
$2$.

Thus, $\optn(I'') \geq \eps^2/2 \cdot
(1-\eps^3)\sum_{i=0}^{1/\eps-1} n_i \geq \frac{13}{27} \cdot
\eps^2 \cdot \sum_{i=0}^{1/\eps-1} n_i$, as by $\eps \leq 1/3$, it
holds that $1-\eps^3 \leq \frac{26}{27}$. From this we obtain that
the following holds:
\begin{equation}\eps^2 \cdot \sum_{i=0}^{1/\eps-1} n_i \leq
\frac{27}{13} \cdot \optn(I'') \ .\label{eqqq}\end{equation}

Therefore, $\optn(I') \leq \optn(I'')+ \eps^3
\sum_{i=0}^{1/\eps-1} n_i \leq (1+3\eps) \optn(I'')$, where the
first inequality follows by the upper bound on the cost of an
optimal nice for $I'$ we derived from $\optn(I'')$, i.e.
(\ref{eqq1}), while the second inequality holds by our last bound
on the total size of large items, i.e. (\ref{eqqq}).
\end{proof}

The last lemma shows that it is sufficient to approximate $I'$, i.e., it is sufficient to approximate $\optn(I')$.  This is the goal of the last two steps of the scheme.

\subsection{The configuration IP}
Throughout this  section, we deal with input $I'$. A configuration
is a vector that encodes the packing of one bin (for $I'$). The
intuition is that since we restrict our packings to be nice
solutions, the packing of one bin (called its configuration) is
characterized by the number of items of each group of each
interval (including the exceeding item if it exists), and the
total size of small items of each interval. Observe that this
information allows us to verify that there is at most one
exceeding item by checking that if we delete an item from the last
interval (with respect to the index) for which there is such
non-zero component (of the configuration) then the total size is
below $1$. Second, it allows us to verify for a configuration that
if the total size of all items is at least $1$, for the last
interval where the number of items is positive, this number is
$1$, and the unique item of the last interval is large.  Thus,
these components allow us to check that the conditions of nice
solutions for the given certificate are satisfied by this
configuration.  By limiting the number of configurations we will
obtain our configuration integer program (IP) that can be solved
in polynomial time for fixed values of $\eps$.

For group $G(i,k)$ denote by $s(i,k)$ the common size of the items in this group.

Formally, a configuration of a bin is a vector $C$ consisting of
the following components.  For every interval $(e_i,e_{i+1}]$ and
every group $G(i,k)$ of this interval, we have a component
$C_{i,k}$ denoting the number of (exceeding or non-exceeding)
items of group $G(i,k)$ in the bin. We note that in the case of an
exceeding item, $C_{i,k}=1$ based on the earlier concepts.
Furthermore, for every interval $(e_i,e_{i+1}]$, we have a
component $C_i$ denoting the total size of small items of the
interval rounded down to the next
integer multiple of $\eps^3$. %%%%%%%$\frac{1}{\eps^3}$.
We will show later that this modification of small items and
rounding (down) of the total size of small items works in a sense
that it is enough to approximate $\optn$ well (in a suitable way)
by the solution an IP.

Such a vector $C$ is a feasible configuration if one of the following conditions hold:
\begin{itemize} \item the total size of the items (including small items) is strictly smaller than $1$, i.e., if $$\sum_{i=0}^{1/\eps -1} \left( C_i+ \sum_{k=1}^{1/\eps^3} C_{i,k}\cdot s(i,k) \right)< 1 ,$$
\item or there is a unique item of the maximum interval (possibly
packed into the bin with some items of smaller intervals, and
after removing this item, the total size of remaining item is
strictly smaller than $1$, that is, if there is $i_{max}$ such
that
$$\sum_{i=i_{max}}^{1/\eps-1} \sum_{k=1}^{1/\eps^3} C_{i,k} =1, $$
where $i_{max}$ can be seen as the index of the interval of the
exceeding item (but it is possible that there are other smaller
indices satisfying this property, but they will not satisfy the
second property). By the definition of nice solutions, there may
or may not be other items of smaller intervals packed into the
bin. It is also required that $C_i=0$ for all $i\geq i_{max}$, and
$$\sum_{i=0}^{i_{max} -1} \left( C_i+ \sum_{k=1}^{1/\eps^3}
C_{i,k}\cdot s(i,k) \right) < 1 . $$
 \end{itemize}

Denote by $\C$ the set of feasible configurations.
\begin{lemma}
The number of feasible configurations is at most $O((1/\eps)^{O(1/\eps^4)})$.
\end{lemma}
\begin{proof}
A configuration is a vector with $O(1/\eps^4)$ components, due to
the following. There are $O(\frac 1{\eps})$ intervals, and for
each one there are $O(\frac 1{\eps^3})$ different sizes in $I'$,
each having a separate component, plus one component for small
items. The components for small items are integers in
$[0,1/\eps^3]$, and other components are integers in
$[0,1/\eps^2]$.
\end{proof}

Our IP uses the integral decision variables $x_C$ for all $C\in
\C$. For a fixed solution, every variable $x_C$ represents the
number of bins for each configuration $C\in \C$. There will be no
additional variables, so the dimension of the IP will be a fixed
constant and we will be able to solve it in polynomial time.  For
interval $(e_i,e_{i+1}]$, denote by $\sigma_i$ the total size of
small items of the interval.

The objective is to minimize the number of bins and by grouping
the bins according to configurations this is equivalent to $\min
\sum_{C\in \C} x_C$.   We have two types of constraints, the first
family is that we need to pack all large items, so for every
interval  $(e_i,e_{i+1}]$ and every group $G(i,k)$ we have the
constraint that all items of $G(i,k)$ are indeed packed, so
$\sum_{C\in \C} C_{i,k}\cdot x_C = \eps^3 n_i$ since the number of
items in the group is $\eps^3 n_i$.  The second family of
constraints is that the total size of small items of each interval
which we pack is approximately the total size of small items of
the interval.  Here, the constraint that we introduce will create
some slack in the right hand side, but we will be able to bound
its impact.  Thus, for every interval  $(e_i,e_{i+1}]$ we have the
constraint $\sum_{C\in \C} C_{i}\cdot x_C \geq \sigma_i$. The
description of the IP is completed by the non-negativity
constraints on the variables.  Thus, we solve the following
configuration IP denoted as (ConfIP):

\begin{eqnarray*}
\min & \sum_{C\in \C} x_C & (ConfIP) \\
s.t. & \sum_{C\in \C} C_{i,k}\cdot x_C = \eps^3 n_i & \forall i,k, \\
& \sum_{C\in \C} C_{i}\cdot x_C \geq \sigma_i & \forall i\\
& x_C \geq 0 & \forall C\in \C .
\end{eqnarray*}

Before presenting our post-processing step that receives an
optimal solution for (ConfIP) and constructs a feasible packing of
the items into bins, we first find an upper bound on the cost of
an optimal solution for (ConfIP) using the cost of $\optn(I')$.
\begin{lemma}
There is a feasible solution for (ConfIP) whose cost as a solution for this program is at most $(1+\eps)\optn(I')$.
\end{lemma}
\begin{proof}
Consider the solution $\optn(I')$ which is an optimal nice
solution for the rounded instance $I'$. We show that this nice
optimal solution of $I'$ induces a feasible solution to
(ConfIP). We first define a configuration for
each bin $B$ in $\optn(I')$.  The configuration $C(B)$
corresponding to $B$ is defined as follows.  $C(B)_{i,k}$ is the
number of items of $G(i,k)$ packed into $B$, and to compute
$C(B)_i$ we first compute the total size of small items of
interval $(e_i,e_{i+1}]$ that are packed into $B$ and then round
down this value to an integer multiple of $\eps^3$.
%%%%%$\frac{1}{\eps^3}$.
Our solution for the IP will be based on the configurations
corresponding to the bins with some additional configurations. For
every interval  $(e_i,e_{i+1}]$ we add to the collection of
configurations (one for each bin of $\optn(I')$) $2\eps^3 \cdot
\optn(I')$ configurations that are copies of the configuration
with all components being zero except for the unique component
$C_i$ that equals $\frac{1}{\eps^3}-1$. In
total we add at most $2\eps^2  \cdot \optn(I') \leq \eps
\optn(I')$ configurations. Next we define the vector $x^{\optn}$
by setting for every $C\in \C$, the value $x^{\optn}_C$ to be the
number of times $C$ appears in the collection of configurations we
defined (the collection with the configurations corresponding to
the bins and the additional configurations).

Clearly, the resulting vector is non-negative, and its cost as a
solution for the IP is at most $(1+\eps)\cdot \optn(I')$. Since
every item of each group of large items is indeed packed in
$\optn(I')$ we conclude that  $ \sum_{C\in \C} C_{i,k}\cdot
x^{\optn}_C = \eps^3 n_i$.  Furthermore, for every interval
$(e_i,e_{i+1}]$, every bin $B$ may contain a total size of small
items of the interval that is larger than $C(B)_i$ by at most
$\eps^3$. Thus, by adding configurations of total size of the
small items of this interval of at least $(\frac{1}{\eps^3}-1)
\cdot 2 \eps^3 \optn(I')$, we guarantee that the constraints
$\sum_{C\in \C} C_{i}\cdot x^{\optn}_C \geq \sigma_i$ are also
satisfied. \end{proof}

\subsection{The post-processing step}
Let $x^*$ denote an optimal solution for (ConfIP).  It remains to
show that we are able to construct a feasible packing of the items
with cost at most $(1+\eps)\cdot \sum_{C\in \C} x^*_C +\frac{1}{\eps}$.

We first open $\sum_{C\in \C} x^*_C$ bins where we associate
$x^*_C$ bins with configuration $C$ for every $C\in \C$.
Furthermore, we have additional set of bins where for every $i$ we
open additional $\lceil \eps^2 \cdot \sum_{C\in \C} x^*_C \rceil$ bins
associated with the interval, each of which with up to
$\frac{1}{\eps^2}$ small items of the interval $(e_i,e_{i+1}]$.
In this way we open
at most $\eps\cdot \sum_{C\in \C} x^*_C +\frac{1}{\eps}$ additional bins. Let us consider the packing of items into
these bins.

For every $C\in \C$ the packing of large items into bins associated with $C$ is carried out such that each such bin is allocated $C_{i,k}$ items of group $G(i,k)$ of interval $(e_i,e_{i+1}]$.  Observe that since $x^*$ satisfies the constraint $\sum_{C\in \C} C_{i,k}\cdot x^*_C = \eps^3 n_i$, this allocation of large items into bins associated with configurations allocates all large items.

Next, consider the small items, and for every interval
$(e_i,e_{i+1}]$ we allocate small items of this interval to the
bins associated with the interval as well as to every bin
associated with configurations $C \in \C$ in the following Next
Fit type approach where we pack the small items of the interval,
one by one in increasing order of their indices. We iterate over
the bins associated with $C$, and pack the small items of the
interval $(e_i,e_{i+1}]$ one by one to the current bin associated
with $C$ as long as adding the next item does not exceed the upper
bound of $C_i$ on the total size of small items of the interval
that are allocated to this bin. Recall that $C_i$ was the total
size of small items for interval $i$, round down to an integer
multiple of $\eps^3$. When we are about to pack an item in a way
that exceeds the total size upper bound, we allocate this small
item to one of the bins associated with the interval, and move to
the next bin associated with $C$ (if there is one) or to the next
configuration. We note that if $C_i=0$ then we allocated no small
items of the interval $(e_i,e_{i+1}]$ to bins associated with $C$.
This is done for every interval and we guarantee that the total
size of small items that are packed into a bin does not exceed the
bound defined by the configuration so the packing of such a bin is
feasible. Furthermore, the packing of the bins associated with
intervals is feasible as every bin among the bins associated with
configurations causes at most one small item (with size below
$\eps^2$) of the interval to be packed into a bin associated with
the interval.  Since each bin associated with the interval has
room for $\frac{1}{\eps^2}$ such items, we have enough room for
all the small items. We conclude the following.
\begin{corollary}
There exists a linear time ($O(n)$ time) algorithm that given a solution $x^*$ for (ConfIP), returns a feasible packing of the items into at most $(1+\eps)\cdot \sum_{C\in \C} x^*_C+\frac{1}{\eps}$ bins.
\end{corollary}

Thus, we established our main result in the offline setting as
follows.
\begin{theorem}
Problem \pr\ admits an asymptotic approximation scheme.
\end{theorem}

\section{Online OEBP}
We recall that OEBP is simply the unfair version of OOEBP. We will
briefly discuss the unfair variant where the algorithm processes a
sequence but an optimal solution can reorder. As mentioned
earlier, any upper bound for this model is also an upper bound for
\pr, since an offline algorithm has more power while an online
algorithm has the same power.

For the case with $1$-items, it is known that the tight asymptotic
bound is $2$ \cite{LDY01} (and this is in fact also an absolute
bound). The algorithm is simply Next Fit (NF), that moves to the
next bin when the current one has load of at least $1$ and cannot
receive additional items. In fact this ratio of $2$ is tight also
for the ratio between \pr\ and the unfair model. To see this fact
that is implied by examples in previous work \cite{LDY01,Zhang98},
consider an input that starts with $N$ items of size $1$ followed
by $N(M-1)$ items of size $\frac 1M$ (for integer $N,M>1$).  When
we consider an optimal solution that can reorder the items, we
have a solution with $N$ bins each of which with $M-1$ items of
size $1/M$ followed by an item of size $1$. When we consider an
optimal solution that cannot reorder the items it consists of $N$
bins each of which with one $1$-item and another $\lceil
N(M-1)/M\rceil $ bins each of which with at most $M$ items of size
$1/M$.  The ratio between these costs approaches $2$ when $M$
grows to infinity, and thus an algorithm for \pr\ that is analyzed
with respect to an optimal algorithm that can reorder the items
cannot have (asymptotic or absolute) approximation ratio smaller
than $2$ (this applies not only for online algorithms but also e.g. for
exponential time offline algorithms).

For the case without $1$-items, this algorithm still has
asymptotic competitive ratio $2$ for \pr. Let $N>0$ be a large
integer, let $M=2N$, and consider the following input. There are
$4N$ items in total, where these items have sizes of $1-\frac 1M$
(large items) and of $\frac 1M$ (smaller items), and the sizes are
alternating (so there are $M=2N$ items of each size). NF creates
$2N$ bins, each with two items of different sizes. An optimal
solution (which does not even need to reorder the input) has $N$
bins with two large items, and one bin with all small items.

We define a different algorithm as follows. Apply NF separately on
items of sizes in $(0,\frac 12)$ and on items of sizes in $[\frac
12,1)$. This is a variant of Harmonic, and we call it NF2. It was
studied by Zhang \cite{Zhang98} and we provide a short alternative
proof.

\begin{proposition}
NF2 has an asymptotic competitive ratio of at most $\frac 32$ for
inputs without 1-items.
\end{proposition}
\begin{proof}
We use the standard approach of a weight based analysis for the
proof (see e.g. \cite{LeeLee85}). For every item of size below
$\frac 12$, its weight is equal to its size. For every item of
size $\frac 12$ or more, its weight is $\frac 12$. Thus, the
weight of an item never exceeds its size and never exceeds $\frac
12$.

For a bin of $\OPT$, the weight is at most $\frac 32$. This holds
since the weight of the last item is at most $\frac 12$, and the
other items have total size below $1$.

For every bin of the algorithm, except for possibly two bins, the
total size is at least $1$ and so is the total weight.
\end{proof}

Once again for \pr\ if we compare ourselves to an offline optimal
solution that can reorder the items this ratio of $\frac{3}{2}$ is
tight (for inputs that do not contain $1$-items) as shown by Zhang
\cite{Zhang98}. For a large integer $N>0$, let $M'=2N^2$. Consider
an input with $2N$ items of size $\frac 12$ (large items) followed
by $2(M'-1)N$ items of size $\frac 1{M'}$ (small items). An
optimal solution that can reorder the items consider reordering
where the small items appear before the large items.  It packs
$M'-1$ small items and one large item into every bin, and it has
$2N$ bins. For an algorithm that is constrained to consider the
input sequence and cannot reorder the items, no matter how many of
the large items are packed in pairs and how many are packed alone,
as $M'$ is divisible by $2$, no bin will have a load above $1$, so
it has at least $\frac{2(M'-1)N}{M'}+N=3N-\frac 1N$ bins.

\section{Absolute competitive ratio for online \pr}
In this section we provide a short discussion regarding the
absolute competitive ratio for \pr. It is easily seen that NF has an
absolute competitive ratio of at most $2$ since every bin has a
load of at least $1$ (except for possibly the last bin, which
still has a positive load), while an optimal solution has load
smaller than $2$ for every bin. We show that this is the best
possible ratio.

\begin{proposition}
The absolute competitive ratio of any online algorithm for \pr\ with
$1$-items is at least $2$.
\end{proposition}
\begin{proof}
Assume that there is an algorithm with absolute competitive ratio
$1\leq q <2$. In particular, when an optimal solution has $k$
bins, the algorithm cannot have more than $2k-1$ bins.

The input is as follows. Let $N \geq 3$ be an integer and let
$\eps=\frac 1{2^N}$. For $i=1,2,\ldots,N$, there are four items in
the $i$th batch, arriving in the order they are stated here, where
their sizes are: $\eps$, $2\cdot i\cdot \eps$, $1-2\cdot i\cdot
\eps$, $1$.

An optimal solution has the following properties. Its cost just
before the arrival of the $1$-item of the $i$th batch is at most
$i$. This holds for any $i$ due to the following packing. For
$i=1$, the first three items of the first batch can be packed into
one bin. Consider the case $i>1$. In this case, the first bin
contains the first, third, and fourth items of the first batch.
For $1<j<i$, there is a bin containing the second item of batch
$j-1$ and three items from the $j$th batch, which are the first,
third, and fourth items. The total size of the second item of
batch $j-1$ and the first and third items of batch $j$ is $2(\cdot
(j-1)\cdot \eps)+\eps+1- 2\cdot j\cdot \eps=1-\eps$, and therefore
such a bin is valid. For batch $i$, the second item of batch $i-1$
and the first three items of batch $i$ are packed into one bin.
The total size of the second item of batch $i-1$ and the first two
items of batch $i$ is $2\cdot (i-1)\cdot \eps+\eps+2\cdot i\cdot
\eps=(4\cdot i -1)\cdot \eps < 4N\eps <1$, so this bin is valid as
well.

Next, we prove that the items of each batch are packed into two
bins by the algorithm, and these bins cannot receive additional
items later. Thus, for every batch $j$, we assume that $2j-2$ bins
were already created, and these bins have loads at least $1$, so
batch $j$ must be packed into new bins, and we show that two bins
are created. Since an optimal solution has at most $j$ bins when
the first three items are presented, and there are already $2j-2$
bins used by the algorithm, only one bin can be used by the
algorithm for these three items. After these three items are
packed, their bin has load above $1$ (the load is $1+\eps$), and
the fourth item is packed into another bin, which will have load
$1$ as a result. Thus, after all four items of batch $j$ arrived,
there are two bins created for this batch, both with loads of at
least $1$.

After all items have arrived, an optimal solution has at most
$N+1$ bins (since the last item which is a $1$-item has to be
packed too), while the algorithm has $2N$ bins. Letting $N$ grow
without bound shows that the absolute competitive ratio cannot be
$q$.\end{proof}

\section{An online algorithm for \pr\ for the case with $\boldsymbol{1}$-items}

We will use a method which was used in the past for classic online
bin packing. In this approach, one partitions items to types, and
tries to combine some types of relatively large items with other
types \cite{LeeLee85,RaBrLL89}. The novelty in our method lies in
the adaptation of this idea to inputs where the order of items
matters. In our analysis, we will split the input at a point where
the behavior of the algorithm changes as a result of a different
input type. More specifically, the packing is different when those
large items stop arriving, and the analysis is different too.
While in work for classic bin packing there are two or more
scenarios \cite{LeeLee85,RaBrLL89,BBDEL_ESA18}, where each of them
may happen for some input, here the two scenarios frequently
happen for one input.

We start with defining the algorithm (or actually a class or kind
of algorithms), and elaborate on the analysis later. The algorithm
uses an integer parameter $M \geq 2$. The algorithm classify items
into classes based on the size of the items.  The item classes are
as follows.

\begin{itemize}
\item Items of size $1$, also called $1$-items are class $0$.
\item For $1 \leq i \leq M-1$, class $i$ consists of items of
sizes in $[\frac 1{i+1},\frac 1i)$. Such items are called {\it
regular items}, or regular items of class $i$ if all of them
belong to this class. \item Class $M$ consists of items of sizes
in $(0,\frac 1 M)$. Such items are called {\it tiny items}.
\end{itemize}
For every class $1 \leq i \leq M$, bins for this class will
contain items of the class and possibly also a $1$-item. There may
be two kinds of bins, called {\it large} and {\it small}. For $1
\leq i \leq M-1$, a large bin is planned to have $i+1$ items of
the class, which is always a feasible bin as the total size of $i$ items
of this class is below $1$. A small bin is planned to have $i$
items of the class, and possibly also a $1$-item that arrives
after the $i$ items of the class have already been packed into
this bin. For tiny items, a large bin will have a total size of at
least $1$, and a small bin will have a total size in $[1-\frac
1M,1)$.

For every class $i\geq 1$, there will be at most one large bin and
at most one small bin that did not receive the required total size
or total number of items of class $i$. These bins will be called
{\it active}, and every class may have an active small bin and an
active large bin. Other bins are called {\it inactive}, and there
may be an arbitrary number of inactive small bins and inactive
large bins. The small inactive bins are partitioned into ready
bins, which are bins that did not receive $1$-items, and used
bins, which are bins that each one of them received a $1$-item as
its last item (this was done after the bin became inactive).

There is a parameter $0 \leq \beta_i \leq 1$, which is the
approximate fraction of bins for class $i \geq 1$ (active and
inactive) that are small (and the fraction of large bins is approximately
$1-\beta_i$).

The algorithm is defined as follows, and its action is based on
item classes.
\begin{itemize}
\item When a $1$-item arrives, act as follows. If there is a ready
bin of some class $i \geq 1$, pack the new item into one such bin
(and the bin becomes used). Otherwise, pack it into an empty bin
of class 0.

\item When a regular item of class $i$ arrives, act as follows. If
there is an active bin for this class, pack it into such a bin.
Otherwise, let $n_i$ be the current number of bins for this class
(large and small, all inactive). Let $n^{\ell}_i$ and $n^s_i$
(where $n_i=n^{\ell}_i+n^s_i$) be the numbers of large and small
bins for class $i$, respectively (all numbers are calculated
excluding the new bin that will be opened). If $n^s_i \leq \beta_i
\cdot n_i$, open a new small active bin for class $i$, and
otherwise (in which case $n^s_i > \beta_i \cdot n_i $ and
therefore $n^{\ell}_i = n_i-n^s_i < n_i - \beta_i \cdot n_i
=(1-\beta_i) n_i $) open a new large active bin for class $i$. The
new item is packed into the new bin.

No matter which bin received the item (new or not), if the bin has
its planned number of items ($i+1$ items if it is large, $i$ items
if it is small), define the bin to be inactive, and if it is
small, additionally define it to be ready.

\item When a tiny item arrives, act as follows.  If there is an
active bin for this class, pack it into such a bin. A large active
bin has a total size below $1$ so it can receive a new item, and a
small active bin has a total size below $1-\frac 1M$ so it can
also receive a new item and remain small.

Otherwise, let $n_M$ the current number of bins for this class
(large and small, all inactive). Let $n^{\ell}_M$ and $n^s_M$
(where $n_M=n^{\ell}_M+n^s_M$) be the numbers of large and small
bins for class $M$, respectively.  If $n^s_M \leq \beta_M \cdot
n_M $, open a new small active bin for class $M$, and otherwise
(in which case $n_M^{\ell}  < (1-\beta_M) \cdot n_M $) open a new
large active bin for class $M$. The new item is packed into the
new bin.

No matter which bin received the item, if the bin has total size
of at least $1$ and it is large, or if it has total size above
$1-\frac 1M$ and it is small, define the bin to be inactive, and
if it is small, additionally define it to be ready.
\end{itemize}

In what follows, we consider only algorithms defined in this way.
For the analysis, we would like to split the input $I$ into two
parts $I_1$ and $I_2$ by removing items packed by the algorithm
into a constant number of bins and partitioning the remaining
items. After this removal of some items and partitioning, the
remaining bins of the algorithm will not contain items of both
sub-inputs simultaneously, and as we will base our weights on this
partition, the partitioning property will also be used in the
analysis.

Let $x$ denote the last $1$-item that is packed into a new bin. If
there is no such item, the first part $I_1$ of the input is
defined to be empty. The removal of bins which defines a removal
of items from $I$ is defined as follows. For all bins that are
active at the time of arrival of $x$, remove their items from $I$
(including items packed after the arrival of $x$). Similarly,
remove all bins that are active at termination. For the remaining
items $I'$ define a partition as follows: $I_1$ consists of all
remaining items arriving before $x$ and including $x$, and $I_2$
consists of all other items that were not removed (those arriving
strictly after $x$).

\paragraph{Properties of the partition of the input.}
Since the order of items in $I'$ is the same as their order in
$I$, and packing of $I$ can be used as a packing for $I'$, we
conclude that $\OPT(I') \leq \OPT(I)$. On the other hand, since
there are at most $2M$ active bins at each time, the cost of the
algorithm is at most $4M$ plus the number of remaining bins.

\begin{claim}
Items of $I_1$ and items of $I_2$ are packed into different bins
by the algorithm. Bins containing these items are not active at
termination.
\end{claim}
\begin{proof}
We start with the first part. Assume by contradiction that this is
not the case, and there is an item $y$ of $I_2$ packed with an
item $z$ of $I_1$. As all items of $I_2$ arrive after all items of
$I_1$, $y$ is packed after $z$. Thus, at the time of arrival of
$x$ the bin of $z$ has to be active. However, all active bins of
this time were removed to obtain $I'$, a contradiction.

Assume by contradiction that there is an active bin. It cannot be
of $I_2$ as bins that are active at termination were removed. It
cannot be of $I_1$ since at the time of arrival of $x$ all items
of $I_1$ already arrived, and active bins were removed then too.
\end{proof}

\begin{claim}
At the time of arrival of $x$, there are no ready bins.
\end{claim}
\begin{proof}
This holds since $x$ is packed into a new bin.
\end{proof}

\begin{claim}
Let $1 \leq i \leq M$. Let $n_i(1)$, $n_i(2)$, $n^{\ell}_i(1)$,
$n^{\ell}_i(2)$, $n^{s}_i(1)$, and $n^{s}_i(2)$ denote the total
number of bins for class $i$ out of bins of $I_1$, the total
number of bins for class $i$ out of bins of $I_2$, the number of
large bins for class $i$ out of bins of $I_1$, the number of large
bins for class $i$ out of bins of $I_2$, the number of small bins
for class $i$ out of bins of $I_1$, and the number of small bins
for class $i$ out of bins of $I_2$, respectively.

Then, it holds that $$(1-\beta_i) \cdot n_i(1) -1  \leq
n^{\ell}_i(1) \leq (1-\beta_i)\cdot n_i(1) + 1, \ \ \ \  \beta_i
\cdot n_i(1) -1 \leq n^{s}_i(1) \leq \beta_i \cdot n_i(1) + 1 \
,$$
$$(1-\beta_i) \cdot n_i(2) - 3 \leq n^{\ell}_i(2) \leq
(1-\beta_i)\cdot n_i(2) + 3, {\mbox{ \ \ and \ \ }} \beta_i \cdot
n_i(2) -3 \leq n^{s}_i(2) \leq \beta_i \cdot n_i(2) + 3 \ .
$$
\end{claim}
\begin{proof}
We will prove the properties for the cases where $\beta_i \in
(0,1)$, since the cases $\beta_i\in\{0,1\}$ are trivial, since
$\beta_i=0$ means that there are no small bins, while $\beta_i=1$
means that there are no large bins.

Consider the opening time of the last large bin for class $i$ in
$I_1$, and let $N_1$ be the number of bins of class $i$ before it
is opened. Since at this moment, a new bin is being opened for
class, at this time there are no active bins for class $i$, and
none of the existing bins for this class will be removed in the
process of moving from $I$ to $I'$. Therefore before the bin is
opened, the number of large bins for class $i$ is exactly
$n^{\ell}_i(1)-1$, and $n_i(1) \geq N_{1} +1$, since there will be
exactly one additional large bin for class $i$ for $I_1$, and
possibly small bins. By $n^{\ell}_i(1) -1 \leq (1-\beta_i)N_1$
(which is the opening rule), we have $n^{\ell}_i(1) \leq
(1-\beta_i)\cdot N_1+1 \leq (1-\beta_i) \cdot (n_i(1)-1)+1 \leq
(1-\beta_i) \cdot n_i(1)+1$, where the second inequality is
because of the value of $N_1$ stated above. Similarly, we can get
$n^{s}_i(1) \leq \beta_i \cdot n_i(1)+1$. If there are no large
bins, or no small bins for class $i$ and $I_1$, the inequalities
hold trivially. As $n_i(1)=n^{s}_i(1)+n^{\ell}_i(1)$, we have
$n^{\ell}_i(1) \geq n_i(1)-(\beta_i \cdot n_i(1)+1)=(1-\beta_i)
n_i(1)-1$ and similarly, $n^{s}_i(1) \geq n_i(1)-((1-\beta_i)
\cdot n_i(1)+1)=\beta_i \cdot n_i(1)-1$.

Now, consider the last large bin for class $i$ that is opened for
$I_2$ and let $N_2$ be the number of bins of class $i$ before it
is opened (including bins of $I_1$ and removed bins). We still
assume here that $0 < \beta_i < 1$, and we assume that there is at
least one large bin for proving an upper bound on such bins, and
that there is at least one small bin for proving an upper bound on
such bins, since the bounds hold trivially otherwise. There are no
active bins for class $i$ at this time, and out of such bins at
most two existing bins for class $i$ will be removed (those that
were active when $x$ arrived) to obtain $n_i(1)+n_i(2)$ bins of
class $i$ at termination (active bins that might be removed at
termination do not exist yet). On the other hand, one large bin
for class $i$ will be created. Thus, $n_i(1)+n_i(2) \geq N_2-1$.
By the opening rule, the number of large class $i$ bins at this
time is at most $(1-\beta_i)\cdot N_2$, and the final number is at
most $(1-\beta_i)(n_i(1)+n_i(2)+1)+1 \leq
(1-\beta_i)(n_i(1)+n_i(2))+2$. Similarly, for small class $i$
bins, the final number is at most $\beta_i \cdot
(n_i(1)+n_i(2))+2$. Given the lower bounds on the number of bins
for $I_1$, the bin numbers for $I_2$ are at most $(1-\beta_i)
\cdot n_i(2)+3$ and $\beta_i \cdot n_i(2)+3$, respectively. Thus,
for $I_2$, since the sum of numbers is $n_i(2)$, the numbers are
at least $(1-\beta_i) \cdot n_i(2)- 3$ and $\beta_i \cdot n_i(2)-
3$, respectively.
\end{proof}

\begin{claim}
For every small bin of $I_1$ for some class $i$, the bin is used.
\end{claim}
\begin{proof}
When the last item $x$ of $I_1$ arrives, it is a $1$-item packed
into a new bin, and it is not removed. Since $x$ cannot be packed
into a ready bin, there are no ready bins for $I_1$, at the
termination of $I_1$.
\end{proof}

\begin{claim}\label{noready}
There are no bins of $I_2$ with a single item that is a $1$-item.
\end{claim}
\begin{proof}
All items of $I_2$ arrived after $x$, which is the last $1$-item
to be packed into a new bin. This holds if $x$ does not exist as
well, because in that case no $1$-item of $I$ is packed into a new
bin.
\end{proof}

\paragraph{The weight functions.}
We define two sets of weights, denoted by $w$ and $v$, where
$v,w:(0,1] \rightarrow \mathbb{R}$, one for $I_1$ and one for
$I_2$. Letting $W$ be the total weight of items of $I_1$ according to $w$  and
letting  $V$ be the total weight of items of $I_2$ according to $v$, we will show
that the cost of the algorithm is at most $W+V+16\cdot M$. We will show this claim by considering the non-active bins of $I'$ and show that (almost) every bin
has total weight at least $1$ according to the suitable weight
function. For an optimal solution, we will find a value $R$ such
that no bin has total weight above $R$ where we define a weight
function $f:I' \rightarrow \mathbb{R}$ for which $f(a)=w(a)$ if $a
\in I_1$ and $f(a)=v(a)$ if $a \in I_2$. Since $V+W  \leq R \cdot
OPT(I')$, the upper bound on the asymptotic competitive ratio will follow.

Function $w$ is defined as follows.

\begin{itemize}
\item The weight of a $1$-item is $1$.

 \item The weight of any item of class $1
\leq i \leq M-1$ is $\frac{1-\beta_i}{i+1-\beta_i}$.

\item The weight of any item of class $M$ of size $\rho$ is
$\frac{1-\beta_M}{1-\beta_M/M}\cdot \rho$.
\end{itemize}

Function $v$ is defined as follows.
\begin{itemize}
\item The weight of a $1$-item is $0$.

\item The weight of any item of class $1 \leq i \leq M-1$ is
$\frac{1}{i+1-\beta_i}$.

\item The weight of any item of class $M$ of size $\rho$ is
$\frac{1}{1-\beta_M/M} \cdot \rho$.
\end{itemize}

We have $ALG(I) \leq n_1+n_2+4M$, where $n_j$ is the number of
bins for $I_j$ for $j=1,2$, and we show $n_1 \leq W+3M$ and $n_2
\leq V+9M$.

\begin{claim}
The total weight of bins of the algorithm for $I_1$ with respect
to $w$ is at least $n_1-3M$.
\end{claim}
\begin{proof}
First, consider bins with $1$-items packed into new bins. Every
such bin has weight $1$. Every remaining bin belong to a class
$i$, where $0< i \leq M$.

Consider a class $1 \leq i\leq M-1$. Recall that by Claim
\ref{noready} there are no ready bins, so every small bin has a
$1$-item. The total weight including $1$-items and the complete
numbers of items of class $i$ (as there are no active bins) is
$\frac{1-\beta_i}{i+1-\beta_i}((i+1)n^{\ell}_i(1)+i\cdot
n^{s}_i(1))+ n^{s}_i(1) =
\frac{(i+1)(1-\beta_i)}{i+1-\beta_i}n^{\ell}_i(1)+
\frac{i(1-\beta_i)+i+1-\beta_i}{i+1-\beta_i} n^{s}_i(1) \geq
\frac{(i+1)(1-\beta_i)}{i+1-\beta_i} ((1-\beta_i) \cdot n_i(1) -1)
+ \frac{i(1-\beta_i)+i+1-\beta_i}{i+1-\beta_i} (\beta_i \cdot
n_i(1) -1) \geq \frac{(i+1)(1-\beta_i)}{i+1-\beta_i} n_i(1)+
\frac{i\beta_i}{i+1-\beta_i}n_i(1)-\frac{2(i+1)(1-\beta_i)+i}{i+1-\beta_i}
\geq n_i(1)-3$, since
$\frac{2(i+1)(1-\beta_i)+i}{i+1-\beta_i}=2+\frac{i(1-2\beta_i)}{i+1-\beta_i}
\leq 3$.

Finally, consider class $M$. The total weight is at least

$$\frac{1-\beta_M}{1-\beta_M/M}((n^{\ell}_M(1)+(1-1/M)\cdot
n^{s}_M(1))+ n^{s}_M(1)$$

$$= \frac{1-\beta_M}{1-\beta_M/M}\cdot n^{\ell}_M(1)+
\frac{(1-\beta_M)(1-1/M)+1-\beta_M/M}{1-\beta_M/M} n^{s}_M(1)$$

$$\geq \frac{1-\beta_M}{1-\beta_M/M}\cdot ((1-\beta_M)n_M(1)-1)+
\frac{(1-\beta_M)(1-1/M)+1-\beta_M/M}{1-\beta_M/M} (\beta_M\cdot
n_M(1)-1)$$

$$=\frac{1-\beta_M}{1-\beta_M/M}\cdot
n_M(1)+\frac{\beta_M(1-1/M)}{1-\beta_M/M}\cdot
n_M(1)-\frac{1-\beta_M+(1-\beta_M)(1-1/M)+1-\beta_M/M}{1-\beta_M/M}
$$

$\geq n_M(1)-3$.  \end{proof}

\begin{claim}
The total weight of bins of the algorithm for $I_2$ with respect
to $v$ is at least $n_2-9M$.
\end{claim}
\begin{proof}
By definition of $x$, the bins of $I_2$ do not include bins where
there is only a $1$-item. Other bins may contain such items, but
their weights are equal to $0$, so we do not discuss such items.

Consider a class $1 \leq i\leq M-1$. The total weight is
$\frac{1}{i+1-\beta_i}((i+1)n^{\ell}_i(2)+i\cdot n^{s}_i(2)) \geq
\frac{i+1}{i+1-\beta_i}((1-\beta_i) \cdot n_i(2) -3)+
\frac{i}{i+1-\beta_i} \cdot (\beta_i \cdot n_i(2) -3)  \geq
n_i(2)-9$.

Finally, consider class $M$.  Recall that $M\geq 2$ and $\beta_M \leq 1$, and thus $\frac{1}{1-\frac{\beta_M}{M}} \leq 2$. The total weight is at least
$\frac{1}{1-\beta_M/M}((n^{\ell}_M(1)+(1-1/M)\cdot n^{s}_M(1))
\geq \frac{1}{1-\beta_M/M}\cdot
((1-\beta_M)n_M(2)-3)+(1-1/M)\cdot(\beta_M\cdot n_M(2)-3))\geq
n_M(2)-9$.
\end{proof}

The next claim holds using the definition of weights directly.
\begin{claim}
For any $1 \leq i \leq M-1$, both weights are in $[0,\frac 1i]$.
\end{claim}

The next claim show that a bin does not contain an item of $I_1$
and an item of $I_2$ at the same time.

\begin{claim}
It is sufficient to analyze bins of $\OPT$ containing items of
only one of the sets $I_1$ and $I_2$.
\end{claim}
\begin{proof}
For every bin of $\OPT$, if the bin has a $1$-item of weight $0$,
we can remove it from the bin for the calculation without changing
the total weight for the bin. Now all $1$-items included in bins
are those that arrived not later than $x$.

We consider two cases.  If the bin has a $1$-item, then this item
has to be the last item of the bin, based on the definition of
\pr. Therefore, the bin only has items that arrived not later than
$x$.

If the bin does not have a $1$-item, we can analyze it according
the weight function $v$, as for any item that is not a $1$-item,
the weight according to $v$ is not smaller than that of $w$, and
this result in the same or a larger value of $R$ for the
algorithm, so it can be assumed that it only has items arriving
strictly after $x$. That is,  it is analyzed as if only items of
$I_2$ were packed into it originally.
\end{proof}

\subsection*{A bad example for algorithms of this class}

A simple example showing that an algorithm of the type we consider
cannot have a very small asymptotic competitive ratio, i.e., a
lower bound is given for this kind of algorithms. Let $M>500$.
There will be no items smaller than $\frac 1{M}$ and for
simplicity we use parameters $\gamma_i$. The value $\gamma_i$
denotes for class $i$ the fraction of items of small bins so
$\gamma_i=\frac{i\cdot \beta_i}{(i+1)\cdot(1-\beta_i)+i\cdot
\beta_i}$. We ignore rounding issues in this example that can be
assumed by using large enough value of $N$ where $N>0$ be a very
large integer.

%Note that increasing the value of $M$ does not harm the algorithm
%in terms of its asymptotic competitive ratio, and it may only
%increase the additive constant.

% 000000000000000000000000000000000000000000
% HOW DO WE JUSTIFY THE LARGE VALUE OF M ??
% 000000000000000000000000000000000000000000

%Before we proceed, consider first the case $\gamma_1 \geq 0.4$.
%Consider an input consisting of $N$ items of size $\frac 13$ and
%$2N$ items of size $\frac 12$. For this input, $\OPT=N$. The
%algorithm will create:
%
%$(\gamma_{2} \cdot N)/2$ small bins of class $2$ and
%$((1-\gamma_{2}) \cdot N)/3$ large bins of this class.
%
%and $(\gamma_{1} \cdot 2N)$ small bins of class $1$ and
%$((1-\gamma_{1}) \cdot 2N)/2$ large bins of this class.
%
%The total number of bins is $\frac N6 \cdot (3\gamma_2
%+2(1-\gamma_2)+12\gamma_1+6(1-\gamma_1))=\frac
%N6(8+\gamma_2+6\gamma_1)\geq \frac N6 \cdot 10.4 > 1.73 \cdot N$.
%
%If $\gamma_1 < 0.4$, we construct another input.

The input is as follows, consisting of $11$ batches.

\begin{itemize}
\item $\frac N{462}$ items of size $\frac{1}{463}$.

\item $\frac N{21}$ items of size $\frac{1}{22}$.

\item $\frac N{3}$ items of size $\frac{1}{7}$.

\item $N$ items of size $\frac 13$.

\item $N$ items of size $\frac 12$.

\item $N$ $1$-items.

\item $N$ items of size $\frac{1}{463}$.

\item $N$ items of size $\frac{1}{22}$.

\item $2N$ items of size $\frac{1}{7}$.

\item $2N$ items of size $\frac 13$.

\item $N$ items of size $\frac 12$.
\end{itemize}

To obtain an offline packing, pack items of the first three
batches into different bins (i.e., one item from the union of
these three batches into each bin), and add one item of each of
the next three batches into these bins. As not all items of
batches 4 an 5 were packed (only $\frac{N}{462}+\frac{N}{21}+\frac
N3$ bins were used), the remaining items are packed into bins with
one item of each size. Before packing the $1$-item, every bin has
a total size below $1$.

For the five last batches, every bin will have one item of size
$\frac 1{463}$, one item of size $\frac1{22}$, two items of size
$\frac 17$, two items of size $\frac 13$, and one item of size
$\frac 12$.

We find $\OPT \leq 2N$.

Consider the action of the algorithm. For the five first batches,
it has the following bins:

\begin{itemize}
\item $(\gamma_{462} \cdot \frac N{462})/462$ small bins of class
$462$ and $((1-\gamma_{462}) \cdot \frac N{462})/463$ large bins
of this class.

\item $(\gamma_{21} \cdot \frac N{21})/21$ small bins of class
$21$ and $((1-\gamma_{21}) \cdot \frac N{21})/22$ large bins of
this class.

\item $(\gamma_{6} \cdot \frac N{3})/6$ small bins of class $6$
and $((1-\gamma_{6}) \cdot \frac N{3})/7$ large bins of this
class.

\item $(\gamma_{2} \cdot N)/2$ small bins of class $2$ and
$((1-\gamma_{2}) \cdot N)/3$ large bins of this class.

\item $(\gamma_{1} \cdot N)$ small bins of class $1$ and
$((1-\gamma_{1}) \cdot N)/2$ large bins of this class.
\end{itemize}

%We claim that the number of small bins is at most $N$. (On second
%thought if it is large it is also fine, so forget it, and we do
%not need the simple example.)

The number of bins after the first six batches are presented is
$\max\{N,(\gamma_{462} \cdot \frac N{462})/462+\gamma_{21} \cdot
\frac N{21})/21+(\gamma_{6} \cdot \frac N{3})/6+(\gamma_{2} \cdot
N)/2+(\gamma_{1} \cdot N)\}+((1-\gamma_{462}) \cdot \frac
N{462})/463+((1-\gamma_{21}) \cdot \frac N{21})/22+((1-\gamma_{6})
\cdot \frac N{3})/7+((1-\gamma_{2}) \cdot N)/3+((1-\gamma_{1})
\cdot N)/2 \geq N+((1-\gamma_{462}) \cdot \frac
N{462})/463+((1-\gamma_{21}) \cdot \frac N{21})/22+((1-\gamma_{6})
\cdot \frac N{3})/7+((1-\gamma_{2}) \cdot N)/3+((1-\gamma_{1})
\cdot N)/2 $.

For the second part of the input, the next bins are built:

\begin{itemize}
\item $(\gamma_{462} \cdot N)/462$ small bins of class $462$ and
$((1-\gamma_{462}) \cdot  N)/463$ large bins of this class.

\item $(\gamma_{21} \cdot N)/21$ small bins of class $21$ and
$((1-\gamma_{21}) \cdot N)/22$ large bins of this class.

\item $(\gamma_{6} \cdot 2N)/6$ small bins of class $6$ and
$((1-\gamma_{6}) \cdot 2N)/7$ large bins of this class.

\item $(\gamma_{2} \cdot 2N)/2$ small bins of class $2$ and
$((1-\gamma_{2}) \cdot 2N)/3$ large bins of this class.

\item $(\gamma_{1} \cdot N)$ small bins of class $1$ and
$((1-\gamma_{1}) \cdot N)/2$ large bins of this class.
\end{itemize}

The numbers of large bins for the different classes are:

\begin{itemize}
\item For class $462$, $((1-\gamma_{462}) \cdot
N)/463+((1-\gamma_{462}) \cdot \frac N{462})/463=((1-\gamma_{462})
\cdot N)/462$.

\item For class $21$, $((1-\gamma_{21}) \cdot
N)/22+((1-\gamma_{21}) \cdot \frac N{21})/22=((1-\gamma_{21})
\cdot N)/21$.

\item For class $6$, $((1-\gamma_{6}) \cdot
2N)/7+((1-\gamma_{6}) \cdot \frac N{3})/7=((1-\gamma_6)\cdot
N)/3$.

\item For class $2$,  $((1-\gamma_{2}) \cdot
2N)/3+((1-\gamma_{2}) \cdot N)/3=(1-\gamma_2)\cdot N$.

\item For class $1$,  $((1-\gamma_{1}) \cdot
N)/2+((1-\gamma_{1}) \cdot N)/2=(1-\gamma_1)\cdot N$.
\end{itemize}

The numbers of bins for the different classes are (excluding small
bins of the first part of the input), in addition to the $N$ bins containing $1$-item.

\begin{itemize}
\item For class $462$, $((1-\gamma_{462}) \cdot
N)/462+(\gamma_{462} \cdot N)/462=N/462$.

\item For class $21$, $((1-\gamma_{21}) \cdot N)/21+(\gamma_{21}
\cdot N)/21=N/21$.

\item For class $6$, $((1-\gamma_6)\cdot N)/3+(\gamma_{6} \cdot
2N)/6=N/3$.

\item For class $2$, $(1-\gamma_2)\cdot N+(\gamma_{2} \cdot
2N)/2=N$.

\item For class $1$,  $(1-\gamma_1)\cdot N+(\gamma_{1} \cdot
N)=N$.
\end{itemize}

In total we have $3N+N/3+N/21+N/462=1563N/462$.

The ratio for the algorithms in the examined class is at least
$\frac{1563}{924}\approx 1.691558441558442$.

To obtain a set of parameters and a tight example, we define the following sequence:
$t_1=22$, $t_{i+1}=t_i(t_i-1)+1$ (so $t_2=463$, $t_3=213907$
etc.). We have $\sum_{i=1}^{\infty} \frac {1}{t_i} < \frac 1{21}$,
and we let $C=\sum_{i=1}^{\infty} \frac {1}{t_i-1} \approx
0.04978822$.

We also let $R=\frac 53+\frac C2 \approx 1.691560779$ (where $R >
1.691560779$ and $R<1.69156078$).

The example above can be modified to give a lower bound of $R$. In
both parts of the input, instead of items of sizes
$\frac{1}{463}=\frac{1}{t_2}$ and $\frac{1}{22}=\frac{1}{t_1}$,
there will be items of sizes $\frac{1}{t_g}$,
$\frac{1}{t_{g-1}},\ldots,\frac{1}{t_1}$, for a fixed integer $g
\geq 3$, where $M>\frac{1}{t_g}$. The numbers of these items are
$N$ in the second part of the input (after the $1$-items arrive),
and $\frac{N}{t_i-1}$ for items of size $\frac{1}{t_i}$ in the
beginning of the input. Since $C<0.05$, it is still possible to
pack $N$ bins by assigning the items of the first $g+1$ batches of
sizes  $\frac{1}{t_g}$,
$\frac{1}{t_{g-1}},\ldots,\frac{1}{t_1},\frac 17$ into different
bins, and pack a triple of items of sizes $\frac 13$, $\frac 12$,
and $1$ into these bins and new bins, such that $N$ bins are
created. Since $\sum_{i=1}^{\infty} \frac {1}{t_i} < \frac 1{21}$,
the second part of the input can be still packed into $N$ bins by
an offline solution.

For items of size $\frac 1{t_i}$, the number of large bins for the
first part of the input is
$(1-\gamma_{t_i-1})\frac{N}{t_i-1}/t_i$. The number of bins for
the second part of the input is $\gamma_{t_i-1}\cdot
\frac{N}{t_i-1}+(1-\gamma_{t_i-1})\frac{N}{t_i}$. In total we have
$$\frac{N(1-\gamma_{t_i-1}+t_i\gamma_{t_i-1}+(t_i-1)(1-\gamma_{t_i-1}))}{t_i(t_i-1)}=\frac{N}{t_i-1}\ .$$
Thus, while the bound on the optimal cost is unchanged, we can
replace $3N+N/3+N/21+N/462$ with $\frac{10N}3+\sum_{i=1}^{g} \frac
{N}{t_i-1}$ in the cost of the algorithm. Letting $g$ grow without
bound, we get $N(\frac{10}3+C)=N \cdot 2R$, for a lower bound of
$R$ on the asymptotic competitive ratio of the above class of
algorithms.

\subsection*{The best possible algorithms of this class}

The set of parameters we will use for our algorithm is as follows.

We let $M$ be a large integer.  We exhibit an algorithm for every
such value of $M$ whose asymptotic competitive ratio tends to $R$
as $M$ grows unbounded. Here $R$ is the same as in the previous
section, i.e., in the bad example. We also use
$\beta_1=\frac{58}{529}$, and
$\beta_2=3-\frac{4-2\beta_1}{(3-R)(2-\beta_1)-1} \approx
0.434052654632836$, $\beta_3=\frac{148}{287}$,
$\beta_4=\frac{15}{23}$, $\beta_5=\frac{13}{23}$, and $\beta_i=1$
for $i \geq 6$.

We have
$\frac{1-\beta_1}{2-\beta_1}+\frac{1-\beta_2}{3-\beta_2}=2-\frac{1}{2-\beta_1}-\frac{2}{3-\beta_2}=2-\frac{1}{2-\beta_1}-2
\cdot \frac{(3-R)(2-\beta_1)-1}{4-2\beta_1}=R-1$.

Recall that $w_i=(1- \beta_i)/(1+ i - \beta_i)$. The weights
according to $w$ for classes $1,2,3,4,5$ are  not larger than
$0.471$, $0.22056078$, $0.139$, $0.08$, $0.08$, and for $6 \leq i
\leq M$, the weight is $0$.

The weights according to $v$ for classes $1,2,3,4,5$ are not
larger than: $0.529$, $0.389719611$, $0.287$, $0.23$, $0.184$,
respectively. For class $6 \leq i \leq M-1$, the weight according
to $v$ is $\frac{1}i$, and for class $M$ it is $\frac{M}{M-1}$
times the size.

For $1 \leq i \leq M-1$, let $v_i$ and $w_i$ denote the weight of
an item of class $i$, based on the weight functions $v$ and $w$,
respectively. We will analyze the possible total weights of packed
bins.

\begin{claim}
Every bin of $\OPT$ with items of $I_1$ has weight not larger than
$R$ according to the weight function $w$.
\end{claim}
\begin{proof}
A given bin can contain a total size below $1$, and one additional
item. The largest possible weight of the additional item is $1$.

Items of sizes in $(0,1)$ having positive weights  according to $w$ are in fact
items of sizes in $[\frac 16,1)$.

For any item of size $u \in (0,\frac 12)$, we have $w(u) <
\frac{2u}3$, and therefore if there is no item of size in $[\frac
12,1)$, the total weight is at most $\frac 53$.

We are left with the case that there is an item of size in $[\frac
12,1)$ packed into the bin. There can be at most two other items
with positive weights packed into the bin, where if there are two
such items, at least one of them has size strictly below $\frac
14$, and if there is an item of size above $\frac 13$, there is no
second item with a positive weight.

Therefore, if there is also an item of size in $[\frac 13,\frac
12)$, there cannot be another item of size above $\frac 16$, and
the total weight is at most  $1+w_1+w_2=
1+\frac{1-\beta_1}{2-\beta_1}+\frac{1-\beta_2}{3-\beta_2}=R$.

We are left with the case that there is no item of size in $[\frac
13,\frac 12)$. In this case the total weight of additional items
is at most $0.139 + 0.08$, and the total weight for the bin is at
most $1.69 < R$.
\end{proof}

\begin{claim} Every bin of $\OPT$ with items of $I_2$ has weight not larger than $R$
according to weight function $v$.
\end{claim}
\begin{proof}
For an item of size $u$ that belongs to one of the classes
$i=2,3,4,5$, we have $\frac{v(u)}{u} < 1.169158832$, $1.148$,
$1.15$, and $1.104$, respectively. For other values of $i$, this
ratio does not exceed the value $\frac 76$. For $i \geq 6$, and
for any item of size $u$, of class $j \geq i$, it holds that
$\frac{v(u)}{u} \leq \frac {i+1}i$. We call the ratio between
weight and size {\it the density of the item}.

A given bin can contain a total size below $1$, and one additional
item (which is an exceeding item). The largest possible weight of
the additional item is at most $v_1 \leq 0.529$. If there is an
additional such item (where there can be at most one additional
item of this class), the total weight will not exceed $2\cdot
0.529 + 1.169158832 \cdot \frac 12 < 1.6426$.

If no other item except for the additional item has size of at
least $\frac 17$, all items are of classes whose indices are $7$
or larger, the densities are no larger than $\frac 87$, and the
total weight is at most $\frac 87+ 0.529< 1.672$. Thus, we are
left with the case that there is one item (of weight at most
$v_1$), there is at least one item of classes $2,3,4,5,6$, and the
remaining items are of classes $2,3,\ldots$. For classes
$2,3,\ldots,M-1$ the weights are equal for all items of one class,
and therefore we can assume that an item of class $i$ has size
$\frac 1{i+1}$ (the total weight may only increase by decreasing
item sizes like this and possibly filling up the possible
remaining space by possible other items).

Note that the only classes for which the density may be larger
than $1.15$ are $2$ and $6$, and the density for these classes is
below $1.169158832$ and $\frac 76$, respectively. The bin can
contain at most two items of class $2$. We will use the properties
$\frac 23+\frac 37 >1$, $\frac 23+\frac 17 > 0.8$, $\frac 13 +
\frac 57>1$, $\frac 13+\frac 27 < 0.62$, and $\frac 47< 0.58$.

If the total size of items of classes 2 and 6 is not larger than
$0.62$, the total weight is at most $v_1 +0.62 \cdot 1.169158832 +
0.38 \cdot 1.15 <1.69088 < R$.

We consider vectors of length $5$ to denote multisets of items of
classes $2,3,4,5,6$ that may be packed into a common bin. A vector
$(z_2,z_3,z_4,z_5,z_6)$ means that there are $z_j$ items of class
$j$. It is valid if $\sum_{j=2}^6 \frac{z_j}{j+1} <1$. The vector
$(0,0,0,0,0)$ was already considered. We consider all valid
vectors for which the total size of items of classes $2,6$ is
above $0.62$.

For each vector we find the remaining space excluding the already
existing items, and we find an integer $i$ such that all remaining
items are of classes with indices $i$ or larger. The first
component of any valid vector is at most $2$. By the lower bound
of $0.62$ on the total size of items of classes $2$ and $6$, we
conclude the following.
\begin{enumerate} \item If it is equal to
$2$, the last component is $0$, $1$ or $2$. \item If it is equal to $1$,
the last component is $3$ or $4$. \item If it is equal to $0$, the last
component is $5$ or $6$ (since this component is at most $6$ and
$\frac 47 < 0.62$). \end{enumerate}

Taking into account the other three components and the property
that the total size is strictly below $1$, there are 16 suitable
vectors, and we analyze each one of them separately.

\begin{itemize}

\item $(0,1,0,0,5)$. In this case the remaining size is below
$\frac 1{28}$. Thus, all remaining items have sizes below $\frac
1{28}$ and their density is at most $\frac {29}{28}$. We get a
total weight of at most $v_1+0.287+\frac 56+\frac {29}{28} \cdot
\frac 1{28}<1.687$.

\item $(0,0,1,0,5)$. In this case the remaining size is  $\frac
3{35}<\frac 1{11}$. Thus, all remaining items have sizes below
$\frac 1{11}$ and their density is at most $\frac {12}{11}$. We
get a total weight of at most $v_1+0.23+\frac 56+\frac {12}{11}
\cdot \frac 3{35}<1.686$.

\item $(0,0,0,1,5)$. In this case the remaining size is
$\frac{5}{42} < \frac 1{8}$. Thus, all remaining items have sizes
below $\frac 1{8}$ and their density is at most $\frac {9}{8}$. We
get a total weight of at most $v_1+0.184+\frac 56+\frac {9}{8}
\cdot \frac 5{42}<1.681$.

\item $(0,0,0,0,5)$. In this case the remaining size is below
$\frac 2{7}$. All remaining items have sizes below $\frac 1{7}$
and their density is at most $\frac {8}{7}$. We get a total weight
of at most $v_1+\frac 56+\frac {8}{7} \cdot \frac 2{7}<1.689$.

\item $(0,0,0,0,6)$. In this case the remaining size is below
$\frac 1{7}$. All remaining items have sizes below $\frac 1{7}$
and their density is at most $\frac {8}{7}$. If they are in fact
smaller than $\frac 18$, the densities are at most $\frac 98$, and
we get a total weight of at most $v_1+1+\frac {9}{8} \cdot \frac
1{7}<1.69$.

Otherwise, there is an item of class $7$ and weight $\frac 17$,
and since its size is $\frac 18$, the remaining items have total
size below $\frac 1{56}$ and density at most $\frac{57}{56}$. If
the maximum density is in fact at most $\frac{58}{57}$, we get a
total weight of at most $v_1+1+\frac {1}{7}+\frac {58}{57} \cdot
\frac 1{56}<1.6900276$.

Finally, if there is an item of size $\frac 1{57}$ and weight
$\frac 1{56}$, the remaining items have a total size below
$\frac{1}{3192}$ and density at most $\frac{3193}{3192}$, and the
total weight is at most $v_1+1+\frac {1}{7}+\frac{1}{56}+\frac
{3193}{3192} \cdot \frac 1{3192}<1.6900277$.
\end{itemize}

In the next cases, we will use the property $v_1+v_2 \leq
0.529+0.389719611=0.918719611$.

\begin{itemize}
\item $(1,0,0,1,3)$. In this case the remaining size is below
$\frac 1{14}$. Thus, all remaining items have sizes below $\frac
1{14}$ and their density is at most $\frac {15}{14}$. We get a
total weight of at most $v_1+v_2+0.184+\frac 36+\frac {15}{14}
\cdot \frac 1{14}<1.679250223245$.

\item $(1,0,1,0,3)$. In this case the remaining size is below
$\frac 4{105}$. Thus, all remaining items have sizes below $\frac
1{26}$ and their density is at most $\frac {27}{26}$. We get a
total weight of at most $v_1+v_2+0.23+\frac 36+\frac {27}{26}
\cdot \frac 4{105}<1.6882800505605$.

\item $(1,0,0,0,3)$. In this case the remaining size is below
$\frac 5{21}$. First, consider the case where all remaining items
have sizes below $\frac 1{8}$. Their density is at most $\frac
{9}{8}$. We get a total weight of at most $v_1+v_2+\frac 36+\frac
{9}{8} \cdot \frac 5{21}<1.6865767538572$.

Otherwise, there is one item of class $7$ (there can be at most
one such item). Excluding this last item, the remaining total size
is at most $\frac{19}{168}$, and the density for it is at most
$\frac{9}8$. We get a total weight of at most $v_1+v_2+\frac 36+
\frac 17+\frac {9}{8} \cdot \frac {19}{168}<1.6888088967143$.

\item $(1,0,0,0,4)$. In this case the remaining size is below
$\frac 2{21}$. Thus, all remaining items have sizes below $\frac
1{10}$ and their density is at most $\frac {11}{10}$.

If there is no item of size $\frac 1{11}$, the density is at most
$\frac {12}{11}$, and we get a total weight of at most
$v_1+v_2+\frac 46+\frac {12}{11} \cdot \frac 2{21}<1.6893$.

Otherwise, the total size of remaining items in this case is below
$\frac 1{231}$ and the density is at most $\frac{232}{231}$, and
we get a total weight of at most $v_1+v_2+\frac 46+\frac
1{10}+\frac {232}{231} \cdot \frac 1{231}<1.68974$.
\end{itemize}

In the remaining cases we use the property that there is one item
of class $1$ and two items of class $2$. Their total weight is
$v_1+2v_2=\frac{1}{2-\beta_1}+\frac{2}{3-\beta_2}=3-R \leq
1.308439221$. In the last case we use the definition of $R$ but in
the other cases we show a slightly better upper bound.

\begin{itemize}
\item $(2,0,0,0,0)$. In this case the remaining size is below
$\frac 1{3}$. All remaining items have density below $\frac 8{7}$.
We get a total weight of at most $v_1+2\cdot v_2+\frac {8}{7}
\cdot \frac 1{3}<1.6893916019524$.

\item $(2,0,0,1,0)$. In this case the remaining size is below
$\frac 1{6}$. All remaining items have density at most $\frac
{8}{7}$. We get a total weight of at most $v_1+2\cdot
v_2+0.184+\frac {8}{7} \cdot \frac 1{6}<1.6829154114762$.

\item $(2,0,1,0,0)$. In this case the remaining size is below
$\frac 2{15}$. All remaining items have density at most $\frac
{8}{7}$.

If there is no item of size $\frac 18$, the densities are not
larger than $\frac 98$, and we get a total weight of at most
$v_1+2\cdot v_2+0.23+\frac {9}{8} \cdot \frac 2{15}<1.688439222$.

If there is an item of size $\frac 18$, the remaining items have total size below $\frac{1}{120}$ and
density at most $\frac{121}{120}$, and we get a total weight of at
most $v_1+2\cdot v_2+0.23+\frac 17+ \frac {121}{120} \cdot \frac
1{120}<1.689699141635$.

\item $(2,1,0,0,0)$. In this case the remaining size is below
$\frac 1{12}$. All remaining items have density at most $\frac
{13}{12}$. We get a total weight of at most $v_1+2\cdot
v_2+0.287+\frac {13}{12} \cdot \frac 1{12}<1.6857169988$.

\item $(2,0,0,1,1)$. In this case the remaining size is below
$\frac 1{42}$. All remaining items have density at most $\frac
{43}{42}$. We get a total weight of at most $v_1+2\cdot
v_2+0.184+\frac 16+\frac {43}{42} \cdot \frac
1{42}<1.6834823049003$.

\item $(2,0,0,0,1)$. In this case the remaining size is below
$\frac 4{21}$. All remaining items have density at most $\frac
{8}{7}$.

If all other items have densities of at most $\frac 98$, we get a
total weight of at most $v_1+2\cdot v_2+ \frac 16+\frac {9}{8}
\cdot \frac 4{21}<1.6893916019524$.

Otherwise, there is also an item of size $\frac 18$ and weight
$\frac 17$ and the total size of other items is at most
$\frac{11}{168}$, and their densities are at most $\frac
{16}{15}$, and  we get a total weight of at most $v_1+2\cdot
v_2+\frac 16+\frac 17+\frac {16}{15} \cdot \frac
{11}{168}<1.6878043003651$.

\item $(2,0,0,0,2)$.

In this case, we will find the total weight of the largest five
items of the analyzed bin.

We have $\frac{1}{2-\beta_1}+\frac{2}{3-\beta_2}+\frac
26=3-R+\frac 26$.

The analysis of the remaining items is similar to the algorithm
Harmonic Fit \cite{LeeLee85}, and we get from these items a weight
of $C$ as shown below. Since $R=\frac 53+\frac C2$, the total
weight of a bin is at most $\frac{10}3-R+C=R$.

More specifically, let $\updelta=\frac{1}{M(M-1)}$. We show that
the total weight of remaining items is at most $C+\updelta$, so as
$M$ grows to infinity, the bounds tends to $R$.

The total size of these items is at most $\frac 1{21}$. Consider a
non-increasing sorted list of items of sizes not smaller than
$\frac{1}{M}$. Assume that all these items have sizes that are
reciprocals of integers by rounding them down (and keeping the
weights unchanged). Comparing this list to the list
$\frac{1}{t_1}$, $\frac{1}{t_2},\ldots,\frac{1}{t_f}$, where $f$
is the largest integer such that $t_f \leq M$, consider the first
item that is different. We will show and use a certain greedy
reciprocal sum property. If no item is different and the lengths
of the two lists are equal, we have a total weight of at most
$\sum_{i=1}^{f} \frac 1{t_i-1} + \frac{1}{t_{f+1}-1}\cdot
\frac{M}{M-1}=\sum_{i=1}^{f+1} \frac 1{t_i-1} +
\frac{1}{t_{f+1}-1}\cdot \frac{1}{M-1}< C+\frac{1}{M(M-1)}$.
Otherwise, the list may be shorter, or it is possible that there
is a different item instead of $\frac{1}{t_j}$ for some $j$. Since
the remaining total size excluding two items of size $\frac 13$,
two items of size $\frac 17$, and one item of every size
$\frac{1}{t_1},\frac{1}{t_2},\ldots,\frac{1}{t_{j-1}}$ (where this
list could be empty if $j=1$) satisfies that the remaining space
is below $\frac{1}{t_j-1}$, so the largest reciprocal of an
integers that is next in the list can be $\frac{1}{t_j}$. If the
next element is different, it must be smaller.

No matter whether the list is shorter or whether the $j$th item is
smaller than $\frac{1}{t_j}$, the weight of items of the space
smaller than $\frac{1}{t_j-1}$ is at most $\frac{t_j+1}{t_j}$
times their total size. We have a total weight of at most
$\frac{t_j+1}{t_j(t_j-1)} = \frac{1}{t_j-1}+\frac{1}{t_{j+1}-1}$,
since $t_{j+1}=t_j(t_j-1)+1$, and therefore the total weight does
not exceed $\sum_{i=1}^{j+1} \frac{1}{t_i-1} <C$.
\end{itemize}
\end{proof}

We conclude with the following theorem.
\begin{theorem}
Our algorithm with the above discussed type parameter set has an
asymptotic competitive ratio of at most $R$ for online \pr\ (with
1-items).\end{theorem}

\section{An online algorithm without 1-items}
The algorithm works similarly  to the previous one, but for items
of class $1$ the packing is different and they act as the
$1$-items of the previous algorithm. However, if there is no ready
bin, and there is a bin with exactly one such item (packed there
because there was no ready bin at the time of its arrival), the
new item of class $1$ is packed with another such item and not
into an empty bin. Thus, except for possibly one bin, bins
containing only items of class $1$ contain pairs of items, and
every such item can have a weight of $\frac 12$. Here, item $x$ is
defined to be last item of class $1$ assigned into either a new
bin or a bin already containing one such item (i.e., not into a
ready bin). We stress that here, a bin containing the one item of
class $1$ is not defined as ready.

We will use two parameters for the algorithm, denoted by $\theta$
and $\Delta$, and we let $Q=\frac{1}{1-2\theta}$. We fix the exact
values later, but we will ensure that the following properties
will hold: $0.319 \leq \Delta \leq 0.3195$, $0.03852 \leq \theta
\leq 0.03853$, and $1.0834 < Q  < 1.0835$.

Consider the next values. Let $M=100$. Let
$\beta_2=\frac{2-6\Delta}{1-2\Delta}$ (and therefore $0.22   \leq
\beta_2 \leq 0.24$), $\beta_i=2\theta(i+1) \leq 1$ for any integer
$3 \leq i \leq \frac{1}{2\theta}-1$ (in this case, we have
$i+1<2\theta$, i.e. $\frac{1}{i+1}>\frac{1}{2\theta}$). Otherwise
(for $i \leq M$) $\beta_i=1$, where $11<\frac{1}{2\theta}-1<12$.
We will prove that the (asymptotic) competitive ratio for an
algorithm of the type defined above with these parameters is at
most $1.44465$.

The formulas for the weights for the input part $I_2$ (denoted by
$v$) as functions of $\beta_i$ are the same as before for all
classes except for class $1$, for which the weights are equal to
$0$ (and class $0$ does not exist at all). Recall that class $0$
does not exist for the inputs studied in this section, since there
are no $1$-items.

For $I_1$, the weights $w$ of items of class $1$ are $\frac 12$.
For classes $2 \leq i \leq M-1$, instead of the weight
$\frac{1-\beta_i}{i+1-\beta_i}$, the weight is
$\frac{1-\beta_i/2}{i+1-\beta_i}$. For  items of class $M$, the
weight $\frac{1-\beta_M}{1-\beta_M/M}$ times the size is replaced
with $\frac{1-\beta_M/2}{1-\beta_M/M}$ times the size. The
difference is based on the property that the weight of class $1$
items is just $\frac 12$ so a bin that received such an item still
requires a weight $\frac 12$ from the other items.

We get
$w_2=\frac{1-\beta_2/2}{3-\beta_2}=\frac{1-\frac{1-3\Delta}{1-2\Delta}}{3-\frac{2-6\Delta}{1-2\Delta}}=\Delta$
and $v_2=\frac{1}{3-\frac{2-6\Delta}{1-2\Delta}}=1-2\Delta$, where
$0.361 \leq 1-2\Delta\leq 0. 362$.

For $ 3 \leq i \leq \frac{1}{2\theta}-1$ we have
$v_i=\frac{1}{i+1-\beta_i}=\frac{1}{i+1-2\theta(i+1)}=\frac{1}{i+1}\cdot\frac{1}{1-2\theta}=\frac{Q}{i+1}$
(and the density is at most $Q$), and
$w_i=\frac{1-\beta_i/2}{i+1-\beta_i}=\frac{1-\theta(i+1)}{i+1-2\theta(i+1)}=\frac
Q{i+1}-\frac{\theta}{1-2\theta}=\frac{1}{i+1}\cdot(Q-(i+1)\theta/(1-2\theta))=Q\frac{1-(i+1)\theta}{i+1}$
(and the density is at most $Q(1-(i+1)\theta)$).

For $ M-1 \geq i > \frac{1}{2\theta}-1$ we have
$v_i=\frac{1}{i+1-\beta_i}=\frac{1}{i}=\frac{1}{i+1}\cdot\frac{i+1}{i}=
\frac{1}{i+1}\cdot(1+\frac {1}{i}) \leq
\frac{1}{i+1}\cdot(1+\frac{1}{\frac{1}{2\theta}-1})=\frac{1}{i+1}\cdot\frac{1}{1-2\theta}
= \frac{Q}{i+1}$ (for a density of at most $(\frac{i+1}{i} \leq
Q$) and $w_i=\frac{1-\beta_i/2}{i+1-\beta_i}\leq
\frac{1/2}{i+1}\cdot\frac{1}{1-2\theta} = \frac{Q/2}{i+1} \leq
\frac{0.54175}{i+1}$ (for a density of at most $\frac Q2 \leq
0.54175$).
%  I suggest an $=\frac{1/2}{i+1-1} = \frac{1/2}{(i+1)(1-\frac{1}{i+1})}$ term after the second one.}

For class $M$, we have $\frac{1-\beta_M/2}{1-\beta_M/M}=\frac
{50}{99}<0.5051<Q/2$ and $\frac{1}{1-\beta_M/M}=\frac
{100}{99}<1.0102<Q$, for the weights $w$ and $v$, respectively.

For a set of items of classes in $\{j,j+1,\ldots,M\}$, such that
$\frac{1}{2\theta}-1< j <M$, the density is at most
$\frac{j+1}{2j}$ for $w$ and at most $\frac{j+1}j$ for $v$, due to
properties discussed earlier. Similarly, for a set of items of
classes in $\{j,j+1,\ldots,M\}$, such that $3 \leq j \leq
\frac{1}{2\theta}-1$, the density is at most $Q$ for $v$ and at
most $Q(1-(j+1)\theta)$ for $w$, because of earlier properties as
well, and since the function
$Q(1-(i+1)\theta)=\frac{1-(i+1)\theta}{1-2\theta}$ is a
monotonically decreasing function of $i$ (for a fixed value of
$\theta$).

%
%
%
% and $w_6=0.101075$ and $v_6=0.15957$,
%where $\frac{w_6}{1/7}=0.707525$ and $\frac{v_6}{1/7}=1.11699$.
%
%For $i=3,4,5,7,8,9$, we have $w_i=\frac{1-\beta_i/2}{i+1-\beta_i}=
%\frac{1-(i+1)/22}{i+1-(i+1)/11}=\frac{21-i}{20}\cdot \frac 1{i+1}$
%and
%$v_i=\frac{1}{i+1-\beta_i}=\frac{1}{i+1-(i+1)/11}=\frac{11}{10}
%\cdot \frac 1{i+1}$. For classes $i=10,11,\ldots, 99$, we have
%$w_i=\frac{1/2}{i}=\frac{1}{2i}=\frac{i+1}{2i} \cdot \frac
%{1}{i+1}=(\frac 12+\frac{1}{2i}) \cdot \frac {1}{i+1}$ and
%$v_i=\frac 1i = \frac{i+1}{i}\cdot \frac{1}{i+1} \leq
%\frac{11}{10}\cdot \frac{1}{i+1}$. For an item of class $M$, the
%size is multiplied by densities of $\frac {50}{99}$ and $\frac
%{100}{99}$, respectively, to get the weights.

For $I_1$, the largest weight of any item is $\frac 12$, according
to $w$ and for $I_2$, the largest weight of any item according to
$v$ is $1-2\Delta$. These values will be used as upper bounds for
the last items of bins.

We bound the total weight of items of a fixed bin that are not the
last item of the bin (that is, without the exceeding item of the
bin). For $I_2$, the ratio between the weight and the size never
exceeds $Q$ except for items of class $2$. The number of items of
class $2$ is at most two, and we use multiplier $Q$ for other
items, so the largest weight is below $(1-2\Delta)+\max\{ \frac
Q3+2(1-2\Delta),2\cdot\frac Q3+(1-2\Delta), Q \}$. This value is
calculated later using the exact values of $\theta$ (and $Q$) and
$\Delta$.

For $I_1$, we consider several cases.

We start with the case where there is no item of size $\frac 12$
(that is not the last item). In this case there can be at most two
items of size $\frac 13$. We use the property that the density for
all other items is at most
$\max\{Q/2,Q(1-4\theta)\} \leq \max\{0.55,\frac{1-4\theta}{1-2\theta}\}$. The function $\frac{1-4\theta}{1-2\theta}$ is a
monotonically decreasing function of $\theta$ and therefore the
density is at most $0.91653$. For items of size $\frac 13$, the density is at most $3\Delta \leq 0.9585$.
The total weight is at most $0.5+0.9585\cdot\frac 23+0.91653\cdot \frac 13 \leq 1.44451$.

%
%
%If there are no items of size
%$\frac 13$, the total weight is at most $0.91653+1/2<1.42$.
%
%Consider the case that there is one item of size $\frac 13$. The
%density of items of sizes $\frac 15$ or less is at most
%$\frac{1-5\theta}{1-2\theta}\leq 0.8748$. There can be at most two
%items of size $\frac 14$, and we get at most $\frac 12+0.32+
%0.91653\cdot \frac 12+ 0.8748 \cdot \frac 16 < 1.425$.
%
%If there are two such items, consider the remaining items, whose
%total size is below $\frac 13$. If there is no item of size $\frac
%14$ and no item of size $\frac 15$, the density for all other
%items is at most $\frac{1-6\theta}{1-2\theta}\leq 0.847$ for a
%total of $\frac 12+2\cdot 0.3215  +0.847\cdot \frac 13 < 1.43$.
%Otherwise (in addition to the last item of size $\frac 12$), there
%are two items of sizes $\frac 13$, an item of size $\frac 14$ or
%an item of size $\frac 15$, and smaller items in the remaining
%space. In the first option, the remaining items have total size
%below $\frac 1{12}$, and in the second case, the remaining items
%have total size below $\frac {2}{15} < \frac 17$. The densities of
%remaining items do not exceed $0.5784$ and $0.77$, respectively.
%So the total weight is at most $\frac 12+2\cdot
%0.3215+\max\{0.886/4+0.5784/12,0.847/5+0.77\cdot (2/15)\}<1.42$.

Otherwise, there are two items of size $\frac 12$ (the last item
and another item), and the remaining items have total size below
$\frac 12$. If there are no items of size $\frac 13$ and $\frac
14$, the density of remaining items is at most
$Q(1-5\theta)<0.875$ for a total weight of at most $1+0.875 \cdot
\frac 12 < 1.44$. There can be just one item of size $\frac 13$ or
$\frac 14$. If its size is $\frac 14$, the total size of other
items is below $\frac 14$. If all other items have sizes of at
most $\frac 16$, the total weight is at most $1+0.91653\cdot \frac
14+0.8331\cdot \frac 14<1.44$. Otherwise, there is also an item of
size $\frac 15$, and the remaining items have total size below
$\frac{1}{20}$, and sizes of at most $\frac 1{21}$. Their
densities do not exceed $\frac{21}{40}$. The total weight in this
case is at most $1+0.91653\cdot \frac 14+0.875\cdot \frac
15+\frac{21}{40}\cdot \frac 1{20}<1.44$.

We are left with the case that there is also an item of size
$\frac 13$, and the total size of other items is below $\frac 16$.
If there is no item whose size is in the set $\{\frac 17,\frac 18\}$ (i.e., of class $6$ or $7$, where there is at most one such item), the density
of other items is at most $0.707853$, and the total weight is at
most $1+0.3195+0.707853\cdot \frac 16 <1.44$.

%If there is an item of size $\frac 1{10}$, the remaining items
%have total size $\frac 1{15}$ and densities at most $0.54$, for a
%total weight at most $1+0.3215+0.69337\cdot \frac 1{10}+0.54\cdot
%\frac 1{15} <1.43$.
%
%If there is an item of size $\frac 19$, the remaining items have
%total size $\frac 1{18}$ and densities at most $0.54$, for a total
%weight at most $1+0.3215+0.7317\cdot \frac 19+0.54\cdot \frac
%1{18}<1.433$.

If there is an item of size $\frac 18$, the remaining items have
total size below $\frac 1{24}$, so their sizes do not exceed
$\frac 1{25}$,  and the densities are at most $\frac{25}{48}$, for
a total weight at most $1+\Delta+Q\cdot
\frac{1-8\theta}{8}+\frac{25}{48}\cdot \frac{1}{24}< 1.435$.

Finally, we have the case where there is an item of size $\frac
17$,  the remaining total size is below $\frac{1}{42}$, the sizes
do not exceed $\frac{1}{43}$, and densities are at most
$\frac{43}{84}$. The total weight is at most $1+\Delta+Q\cdot
\frac{1-7\theta}{7}+\frac{43}{84}\cdot \frac 1{42}$.

%The specific values we use are $\theta=0.0357742745$,
%$\Delta=0.320489651825869$, and therefore
%$1.07706224049<Q<1.0770622405$. We have $\frac Q3<0.35902074683$
%and $1-2\Delta<0.359020696348263$, so $(1-2\Delta)+\max\{ \frac
%Q3+2(1-2\Delta),2\frac Q3+(1-2\Delta), Q \}<1.4360828863565$. We
%also have $1+\Delta+Q\cdot \frac{1-8\theta}{8}+\frac{25}{48}\cdot
%\frac{1}{24}<$.

The specific values we use are $\theta=0.038526551295994$,
$\Delta=0.319418991002646$, and therefore
$1.0834859544<Q<1.08348595441$.  We
have $1+\Delta+Q\cdot \frac{1-8\theta}{8}+\frac{25}{48}\cdot
\frac{1}{24}<1.435$, and  $1+\Delta+Q\cdot
\frac{1-7\theta}{7}+\frac{43}{84}\cdot \frac 1{42}<1.44465$. We have $\frac Q3<0.3611619848005$
and $1-2\Delta<0.361162018$, so $(1-2\Delta)+\max\{ \frac
Q3+2(1-2\Delta),2\cdot\frac Q3+(1-2\Delta), Q \}<1.44465$.

%1.08348595440148
%objective value:                     1.44464794464794
%R                                    1.44464794464794   (obj:1)
%Q                                    1.08348595440143   (obj:0)
%q                                   0.038526551295994   (obj:0)
%P                                  0.0417429772007168   (obj:0)
%D                                   0.319418991002646   (obj:0)
%A                                    0.01218820861678   (obj:0)
%B                                  0.0217013888888889   (obj:0)

\subsection*{A simple bad example for algorithms of this class}
Next, we show that our analysis of algorithms of this class, without
$1$-items, is almost tight. Let $M>100000$. There will be no
items smaller than $\frac 1{M}$ and for simplicity we use
parameters $\gamma_i$. Similarly to the example for algorithms with $1$-items, the value $\gamma_i$ denotes for class $i$ the
fraction of items of small bins. We ignore rounding issues in this
example. Let $N>0$ be a very large integer.

The input is as follows, consisting of $13$ batches.

\begin{itemize}
\item $\frac {12N}{90901}$ items of size $\frac{1}{90903}$.

\item $\frac {2N}{450}$ items of size $\frac{1}{452}$.

\item $\frac {12N}{300}$ items of size $\frac{1}{302}$.

\item $\frac {80N}{41}$ items of size $\frac{1}{43}$.

\item $16N$ items of size $\frac{1}{7}$.

\item $16N$ items of size $\frac 13$.

\item $32N$ items of size $\frac 12$.

\item $6N$ items of size $\frac{1}{90903}$.

\item $N$ items of size $\frac{1}{452}$.

\item $6N$ items of size $\frac{1}{302}$.

\item $40N$ items of size $\frac{1}{43}$.

\item $40N$ items of size $\frac{1}{7}$.

\item $8N$ items of size $\frac 13$.
\end{itemize}

To obtain an offline packing, pack $16N$ bins, where every bin has
at most one item of the first four batches (such that all items
are packed), one item of size $\frac 17$, one item of size $\frac
13$, and one item of size $\frac 12$. All these bins have total
sizes strictly below $1$, and each one receives another item of
size $\frac 12$. Afterwards, $N$ bins receive (each) one item of
size $\frac 1{452}$, four items of size $\frac 1{43}$, four items
of size $\frac 17$, and one item of size $\frac 13$, for a total
size below $1$. Another $6N$ bins receive (each) one item of size $\frac{1}{90903}$, one item of size
$\frac 1{302}$, six items of size $\frac 1{43}$, and six items of
size $\frac 17$, for a total size below $1$. Each of these $7N$
bins also receives one item of size $\frac 13$ as the last item.

We find $\OPT \leq 23N$.

Consider the action of the algorithm. For the first seven batches,
it has the following bins:

\begin{itemize}
\item $(\gamma_{90902} \cdot \frac {12N}{90901})/90902$ small bins
of class $90902$ and $((1-\gamma_{90902}) \cdot \frac
{12N}{90901})/90903$ large bins of this class.

\item $(\gamma_{451} \cdot \frac N{225})/451$ small bins of class
$451$ and $((1-\gamma_{451}) \cdot \frac N{225})/452$ large bins
of this class.

\item $\gamma_{301} \cdot \frac{N}{25} /301$ small bins of class
$301$ and $(1-\gamma_{301}) \cdot \frac{N}{25} /302$ large bins of
this class.

\item $\gamma_{42} \cdot \frac{80N}{41} /42$ small bins of class
$42$ and $(1-\gamma_{42}) \cdot \frac{80N}{41} /43$ large bins
of this class.

\item $\gamma_{6} \cdot 16N /6$ small bins of class $6$ and
$(1-\gamma_{6}) \cdot 16N/7$ large bins of this class.

\item $\gamma_{2} \cdot 16N /2$ small bins of class $2$ and
$(1-\gamma_{2} )\cdot 16N/3$ large bins of this class.

\end{itemize}

Items of size $\frac 12$ are packed into the existing small bins,
one per bin, and the remaining items are packed in pairs. The
total number of bins created before these items are packed is
below $11N$, and every small bin has items of total size above
$\frac 12$ but below $1$.

%We claim that the number of small bins is at most $N$. (On second
%thought if it is large it is also fine, so forget it, and we do
%not need the simple example.)

After rearranging, we see that the number of bins after the first
seven batches are presented is:
$$\frac 12 \left(\frac{\gamma_{90902} \cdot {12N}}{90901 \cdot 90902}+\frac{\gamma_{451} \cdot N}{225 \cdot 451}+\frac{\gamma_{301} \cdot  N}{25 \cdot 301}+\frac{\gamma_{42}\cdot {80N}}{41 \cdot 42}
+\frac{\gamma_{6} \cdot 16N}6 +\frac{\gamma_{2} \cdot
16N}2\right)+16N$$
$$+\frac{(1-\gamma_{90902}) \cdot {12N}}{90901 \cdot 90903}+\frac{(1-\gamma_{451}) \cdot  N}{225 \cdot 452}+\frac{(1-\gamma_{301}) \cdot  N}{25 \cdot 302}
+\frac{(1-\gamma_{42}) \cdot {80N}}{41\cdot 43}$$
$$+\frac{(1-\gamma_{6}) \cdot 16N}7+\frac{(1-\gamma_{2} )\cdot 16N}3
 \ .  $$

For the second part of the input, the next bins are built:

\begin{itemize}
\item $\gamma_{90902} \cdot  6N/90902$ small bins
of class $90902$ and $(1-\gamma_{90902}) \cdot
6N/90903$ large bins of this class.

\item $\gamma_{451}  N/451$ small bins of class
$451$ and $(1-\gamma_{451}) N/452$ large bins
of this class.

\item $\gamma_{301} \cdot 6N/301$ small bins of class
$301$ and $(1-\gamma_{301}) \cdot 6N/302$ large bins of
this class.

\item $\gamma_{42} \cdot  40N/42$ small bins of class
$42$ and $(1-\gamma_{42}) \cdot 40N/43$ large bins
of this class.

\item $\gamma_{6} \cdot 40N /6$ small bins of class $6$ and
$(1-\gamma_{6}) \cdot 40N/7$ large bins of this class.

\item $\gamma_{2} \cdot 8N /2$ small bins of class $2$ and
$(1-\gamma_{2} )\cdot 8N/3$ large bins of this class.

\end{itemize}

%The numbers of large bins for the different classes are:
When we compute the total number of bins, we can write it down as a weighted sum of the parameters $\gamma_i$ plus some constant term (independent of these parameters).  In this weighted sum the multipliers of $\gamma_i$'s are as follows.

\begin{itemize}

\item The multiplier of $\gamma_{90902}N$ is $\frac{6}{90901\cdot
90902}-\frac{12}{90901\cdot
90903}+\frac{6}{90902}-\frac{6}{90903}=0$.

%%%\frac{6\cdot 90903-12\cdot 90902+6\cdot90901}{90901\cdot90902\cdot90903}=0$.

\item The multiplier of $\gamma_{451}N$ is $\frac{1}{2\cdot
225\cdot 451} -\frac{1}{225\cdot 452}+\frac 1{451}-\frac
1{452}=0$.

\item The multiplier of $\gamma_{301}N$ is $\frac{1}{50\cdot
301}-\frac{1}{25\cdot 302}+\frac 6{301}-\frac 6{302}=0$.

\item The multiplier of $\gamma_{42}N$ is $\frac{40}{41\cdot
42}-\frac{80}{41\cdot 43}+\frac{20}{21}-\frac{40}{43}=0$.

\item The multiplier of $\gamma_{6}N$ is $\frac
43-\frac{16}{7}+\frac{20}3-\frac{40}7=0$.

\item The multiplier of $\gamma_{2}N$ is $4-\frac{16}{3}+4-\frac
83=0$.

\end{itemize}

Thus, we are left with the constant term where we used that it
can be seen after rearranging that the above values are all
zeroes. That is,  there are
$N(16+\frac{12}{90901\cdot 90903}+\frac{1}{225\cdot
452}+\frac{1}{25\cdot 302}+\frac{80}{41\cdot
43}+\frac{16}7+\frac{16}3+\frac{6}{90903}+\frac{1}{452}+\frac{6}{302}+\frac{40}{43}+\frac{40}7+\frac
83)$ bins. The approximate number of bins is $32.9978979841943$,
and the resulting ratio is $1.4346912167$.

%
%\begin{itemize}
%\item For class $462$, $((1-\gamma_{462}) \cdot
%N)/463+((1-\gamma_{462}) \cdot \frac N{462})/463=((1-\gamma_{463})
%\cdot N)/462$.
%
%\item For class $21$, $((1-\gamma_{21}) \cdot
%N)/22+((1-\gamma_{21}) \cdot \frac N{21})/22=((1-\gamma_{21})
%\cdot N)/21$.
%
%\item For class $6$ and $((1-\gamma_{6}) \cdot
%2N)/7+((1-\gamma_{6}) \cdot \frac N{3})/7=((1-\gamma_6)\cdot
%N)/3$.
%
%\item For class $2$ and $((1-\gamma_{2}) \cdot
%2N)/3+((1-\gamma_{2}) \cdot N)/3=(1-\gamma_2)\cdot N$.
%
%
%\item For class $1$ and $((1-\gamma_{1}) \cdot
%N)/2+((1-\gamma_{1}) \cdot N)/2=(1-\gamma_1)\cdot N$.
%\end{itemize}
%
%The numbers of bins for the different classes are (excluding small
%bins of the first part of the input).
%
%
%\begin{itemize}
%\item For class $462$, $((1-\gamma_{463}) \cdot
%N)/462+(\gamma_{462} \cdot N)/462=N/462$.
%
%\item For class $21$, $((1-\gamma_{21}) \cdot N)/21+(\gamma_{21}
%\cdot N)/21=N/21$.
%
%\item For class $6$ and $((1-\gamma_6)\cdot N)/3+(\gamma_{6} \cdot
%2N)/6=N/3$.
%
%\item For class $2$ and $(1-\gamma_2)\cdot N+(\gamma_{2} \cdot
%2N)/2=N$.
%
%
%\item For class $1$ and $(1-\gamma_1)\cdot N+(\gamma_{1} \cdot
%N)=N$.
%\end{itemize}

%%In total we have $3N+N/3+N/21+N/462=1563/462$.

%The ratio is at least $\frac{1563}{924}\approx 1.691558441558442$.

\section{Lower bounds on the asymptotic competitive ratio for \pr}
In this section we improve the known lower bounds on the
asymptotic competitive ratio slightly. The main goal of this
section is to provide complete analytic proofs for these bounds,
as previous work stated them without proving them analytically.
Specifically, the analysis was done using packing patterns, and
the number of such patterns can be large. Since we provide
analytic proof, we can use arbitrarily long sequences. The
improvement results also from using a modified input sequence, as
in \cite{BBG} instead of using a sequence similar to that of
\cite{Vliet92}.

Let $M,N \geq 3$ be integers, where $M$ is divisible by
$(7^{N+1})!$.

An item of type $(i,7)$ for $1 \leq i \leq N$ has size
$\theta_i=\frac{1}{7^i}$. An item of type $j$ for $j=1,2,3$ has
size $\phi_j=\frac 1j$. In other parts of this work items of type
$1$ are called $1$-items (and class $0$),  but for consistency of
the current part we call them items of type $1$ here.

Items are always presented sorted by non-decreasing size and in
batches of $M$ identical items (except for one case where there is
a batch of $2M$ identical items, which can be seen as two batches
with $M$ items each, but we analyze it as a single batch). For $1
\leq k \leq N$, input $I_k$ consists of the first $N-k+1$ batches, where
every batch has $M$ items, and these items are of types $(N,7)$,
$(N-1,7), \ldots, (k,7)$. Input $J_3$ consists of $I_1$ followed
by a batch of $M$ items of type $3$. Input $J_2$ consists of $J_3$
followed by a batch of $M$ items of type $2$. Input $J_{22}$
consists of $J_3$ followed by $2M$ items of type $2$, and this is
the unique case where the batch has $2M$ items. Input $J_1$
consists of $J_2$ followed by $M$ items of size $1$. We say that
input $I_{k_1}$ is a (proper) prefix of input $I_{k_2}$ if $k_2 >
k_1$, and similarly, for every $k$, input $I_k$ is a prefix of
$J_3$, $J_2$, $J_{22}$, and $J_1$, additionally, $J_3$ is a prefix
of $J_2$, $J_{22}$, and $J_1$, and finally $J_2$ is a prefix of
$J_1$.

Input $J_1$ and its prefixes are used for proving a lower bound on
the asymptotic competitive ratio of any algorithm where $1$-items
are possible, and $J_{22}$ and its prefixes are used for the case
without $1$-items.

Let $\Theta_k=\sum_{i=k}^N \theta_i$. We use
$\mu=\theta_N=\Theta_N=\frac{1}{7^N}$, which is the size of the
smallest item. It holds that $\theta_k=7^{N-k} \mu$ for
$k=1,2,\ldots,N$.

We start with finding upper bounds on optimal costs. Let $M'_k=
\frac{M\cdot \Theta_k}{1-\mu+\theta_k}$ for any $1 \leq k \leq N$.

\begin{lemma}
We have $\OPT(I_k) = M'_k$ for any $1 \leq k \leq N$.
\end{lemma}
\begin{proof}
We have
$1-\mu+\theta_k=1-\frac{1}{7^N}+\frac{1}{7^k}=\frac{1}{7^N}\cdot(7^N-1+7^{N-k})=\mu\cdot(7^N-1+7^{N-k})
\geq 1$ and $\Theta_k=\sum_{i=k}^N
\frac{1}{7^{i}}=\frac{1}{7^N}\cdot \sum_{i=k}^N
7^{N-i}=\mu\cdot\sum_{i=0}^{N-k} 7^i $. We have that $M$ is
divisible by $7^N-1+7^{N-k}$, because it is lower than
$7^{N+1}$, and we get that $\frac{M\cdot
\Theta_k}{1-\mu+\theta_k}$, which will be a number of bins,  is an
integer. Since $\frac{1}{\mu}=7^N$, this integer is divisible by
$7^N$.

 We have $\sum_{i=0}^{N-k} 7^i = \frac{7^{N-k+1}-1}6 <
\frac{7^N}6$. Thus, for any $1 \leq k \leq N$,  we get
$M'=\frac{M\cdot \Theta_k}{1-\mu+\theta_k}<\frac{M}6$.

An upper bound on $\OPT(I_k)$ follows from the property that every
bin can contain item of total size strictly below $1$ and one
additional item. The size of the additional item is at most
$\theta_k$. All items have sizes that are integer multiples of
$\mu$, so the total size of a set of items that is strictly below
$1$ is in fact at most $1-\mu$. Next, we will show that it is
possible to produce such a packing.

In the case $k=N$, the claim is $\OPT(I_N)=\frac{M\cdot
\Theta_N}{1-\mu+\theta_k}=M\cdot \mu=\frac{M}{7^N}$, which is
achieved by packing $7^N$ items (each of size $\mu$) into every
bin. To prove the claim for other cases, consider a fixed value $1
\leq k \leq N-1$.

We start the packing of $M'$ bins as follows. Pack six items of
each size $\theta_i$ for every $k+1 \leq i \leq N$ into each bin.
We packed $6M'$ items so far, therefore there are still $M-6M'$
unpacked items for each size $\theta_i$. For the remaining items
of these sizes, create a partition into subsets called blocks. For
every $i=k+1,k+2,\ldots,N$, we create blocks from $M-6M'>0$ items,
so there are separate blocks for  each size $\theta_i$. The number
of items of such a block is $7^{i-k}$. Their total size is
$7^{i-k} \cdot \theta_i = 7^{i-k} \cdot \frac
{1}{7^i}=\frac{1}{7^k}=\theta_k$. The number of blocks is
\begin{equation}(M-6M') \sum_{i=k+1}^N \frac {1}{7^{i-k}}=7^k (M-6M')
\Theta_{k+1}\label{eqq2}\end{equation}

We add $M-M'$ blocks consisting of a single item of size
$\theta_k$, where the remaining $M'$ items of this size will be
the last items of the $M'$ bins. Every bin receives $7^k-1$ blocks
of size $\theta_k$. We show that the total size of items excluding
the last item is strictly below $1$ (it is equal to $1-\theta_N$),
and that the number of packed blocks is exactly the number of
blocks.

Indeed, we have $6\cdot \sum_{i=k+1}^N \frac 1{7^i} + (7^k-1)\cdot
\frac 1{7^k}=6 \cdot \frac{7^{N-k}-1}{6 \cdot
7^N}+1-\frac{1}{7^k}=1-\frac{1}{7^N}$, and the total sizes packed
into bins are as claimed. The number of blocks is $7^k (M-6M')
\Theta_{k+1}+M-M'$, where the first term is exactly
(\ref{eqq2}), and we show that this number of block is equal to $(7^k-1)M'$. This
property is equivalent to
$M'(1+6\Theta_{k+1})=M(\Theta_{k+1}+\frac{1}{7^k})$. Therefore, as
$\Theta_{k+1}+\frac{1}{7^k}=\Theta_k$, and by definition
$M'=\frac{M\Theta_k}{1-\mu+\theta_k}$, we will show that
$1+6\Theta_{k+1}=1-\mu+\theta_k$, or alternatively,
$6\Theta_{k+1}+\mu=\theta_k$ holds for $k<N$. Indeed we have
$6\Theta_{k+1}+\mu= 6(\sum_{i=k+1}^N \frac 1{7^i})+\mu=
7(\sum_{i=k+1}^N \frac 1{7^i})-(\sum_{i=k+1}^N \frac
1{7^i})+\mu=\sum_{i=k+1}^N \frac 1{7^{i-1}}+\mu-\sum_{i=k+1}^N
\frac 1{7^i}=\sum_{i=k}^N \frac 1{7^{i}}-\sum_{i=k+1}^N \frac
1{7^i}=\frac{1}{7^k}=\theta_k$.
\end{proof}

\begin{corollary}
We have $\OPT(I_k) \leq \frac{7M}{6\cdot(7^k+1)}$ for any $1 \leq
k \leq N$.
\end{corollary}
\begin{proof}
We prove that $\frac{\Theta_k}{1-\mu+\theta_k} \leq
\frac{7}{6\cdot(7^k+1)}$ holds.

This is equivalent to $6\Theta_k(7^k+1) \leq
7(1-\frac{1}{7^N}+\frac 1{7^k})$. By $\Theta_k=\frac {1}{7^N}
\cdot \frac{7^{N-k+1}-1}6$, which is the sum of a geometrical
series with common ratio 7, we find that it is required to prove
that $(7^{N-k+1}-1)(7^k+1) \leq 7^{N+1}(1-\frac 1{7^N}+\frac
1{7^k})$. This is equivalent to $7^{N+1}-7^k+7^{N-k+1}-1 \leq
7^{N+1}-7+7^{N-k+1}$ and to $7^k\geq 6$, which holds for any $k
\geq 1$.
\end{proof}

\begin{lemma}
We have $\OPT(J_3) \leq \frac{3M}{8}$, $\OPT(J_2) \leq
\frac{2M}3$, $\OPT(J_{22}) \leq M$, and $\OPT(J_{1}) \leq M$.
\end{lemma}
\begin{proof}
Create blocks of items in the following way. A block has one item
of every size $\theta_k$ for $1 \leq k \leq N$. The total size for
a block is $\sum_{i=1}^N \frac 1{7^i} < \sum_{i=1}^{\infty} \frac
1{7^i} = \frac 16$. The number of blocks is $M$.

For $J_3$, create $\frac{M}{8}$ bins with four blocks defined
above, and $\frac M4$ bins with two blocks. Add one item of size
$\frac 13$ to each bin out of the first $\frac{M}{8}$ bins and two
such items to the last $\frac{M}{4}$ bins. Every bin has total
size below $1$. For every $1 \leq k \leq N$, the number of packed
items of type $(k,7)$ is $4\cdot \frac M8+ 2\cdot \frac M4=M$. For
type $3$ items, the number of packed items is $\frac M8+2\cdot
\frac M4=\frac{5M}8$. The remaining type $3$ items ($\frac {3M}8$
items) are added as last items into the bins. This results in a
valid solution.

%%%A question. Is there no further potential in this?} I could not improve it..

For $J_2$, create $\frac M3$ bins with one block and $\frac M3$
bins with two blocks. The first kind of bins will also have one
item of type $2$ and one item of type $3$. The second kind of bins
will also have two items of type $3$. Every bin has total size below
$1$. All items are packed except for $\frac{2M}3$ items of type
$2$, and these items are packed as the last items for the bins.

For $J_{22}$, create $M$ bins, each containing one block, one item
of type $3$ and one item of type $2$. Every bin has total size
below $1$. All items are packed except for $M$ items of type $2$,
and these items are packed as the last items for the bins. For
$J_1$, the packing is the same as for $J_{22}$ with the difference
that the last item of each bin is of type $1$.
\end{proof}

We assign weights to items. An item of type $j$ for $j=1,2,3$ has
weight $2$. An item of type $(i,7)$ has weight
$\frac{1}{7^{i-1}}$. Let $W_i$ and $V_i$ denote the maximum total
weights of bins containing items of types presented not earlier
than type $(i,7)$ items for $1 \leq i \leq N$, for the case with
and without items of type $1$ (that is, where the complete inputs
are $J_1$ and $J_{22}$, respectively). Let $A_j$ and $B_j$ be the
maximum total weights for bins containing items arriving not
earlier than items of type $j$ for the two cases ($B_1$ is
undefined).

\begin{lemma}
It holds that $A_1=2$, $A_2=B_2=4$, and $A_3=B_3=6$. Additionally,
$W_k=V_k=9-\frac{1}{7^{k-1}}$.
\end{lemma}
\begin{proof}
Consider a bin $B$. The last item never has weight above $2$, and
we find an upper bound on the weight of other items of the bin,
whose total size is strictly below $1$.  Let $S$ be such a subset
of items with total size strictly smaller than $1$ that maximizes
the total weight of its items.

We have $A_1=2$, since in this case $S$ must be empty. Additionally, $A_2=B_2=4$, since
$S$ has at most one item of type $2$ and one such item is a valid $S$. Finally
$A_3=B_3=6$, since $S$ has at most two items and a pair of such items is a valid $S$.

Next, assume that $B$ can have items arriving not earlier than
type $(i,7)$ items. Consider the subset $S$. Every item of size $\frac 13$ or larger can be replaced with
two items of size $\frac 17$ without decreasing the total weight.
Every item of size $\frac 1{7^{i'}}$ with $i'<i$ can be replaced
with $7^{i-i'}$ items of size $\frac 1{7^i}$, having the same total
weight. Thus, we can assume that the subset has only items of size
$\frac{1}{7^i}$. Their number is at most $7^i-1$ and therefore
their weight is at most
$\frac{7^i-1}{7^{i-1}}=7-\frac{1}{7^{i-1}}$ and we note that such number of items of size $\frac{1}{7^i}$ is indeed a valid $S$. The claim follows
from adding the last item of weight at most $2$.
\end{proof}

In the recent work on proving lower bounds for online bin packing
type of problems, it was shown
\cite{BDE,BBDEL_newlb,balogh2017lower,Ep19}  that if we can assign
weights to items as it is done here, then a lower bound on the
asymptotic competitive ratio as follows holds. This lower bound is
defined as a ratio between a given pair of a numerator and a
denominator. The numerator is the total weight of all items while
the denominator is a valid upper bound on the value
$\sum_{i=1}^{N-1}
(W_{i+1}-W_i)OPT(I_{i+1})+(W_1-A_3)OPT(I_1)+(A_3-A_2)OPT(J_3)+(A_2-A_1)OPT(J_2)+A_1\cdot
OPT(J_1)$.

We now compute the total weight of all items in the inputs $J_1$
and $J_{22}$, where these values are equal, i.e., we consider the
numerator of the above ratio. The total weight of items of types
$1,2,3$ is $6M$. The weight of other items is $M \cdot \frac
{7^{N}-1}{6\cdot 7^{N-1}}$ (using the sum of a geometrical series
with common ratio 7).

Next, we consider the denominator.  We have $\sum_{i=1}^{N-1}
(W_{i+1}-W_i)OPT(I_{i+1})+(W_1-A_3)OPT(I_1)+(A_3-A_2)OPT(J_3)+(A_2-A_1)OPT(J_2)+A_1\cdot
OPT(J_1)=\sum_{i=1}^{N-1}
(\frac{1}{7^{i-1}}-\frac{1}{7^{i}})OPT(I_{i+1})+2\cdot
(OPT(I_1)+OPT(J_3)+OPT(J_2)+OPT(J_1)) \leq \sum_{i=1}^{N-1}
(\frac{1}{7^{i-1}}-\frac{1}{7^{i}})\frac{7M}{6\cdot(7^{i+1}+1)}+2M\cdot
(\frac{7}{48}+\frac 38+\frac 23+1)=M((\sum_{i=1}^{N-1}
\frac{1}{7^{2i}+7^{i-1}})+4.375)$.

We find an upper bound on the series $\sum_{i=1}^{N-1}
\frac{1}{7^{2i}+7^{i-1}}$. This is done by calculating the first
six elements and bounding the other ones. We have
$\frac{1}{7^{2i}+7^{i-1}} \leq \frac{1}{7^{2i}}=\frac{1}{49^i}$
and therefore $\sum_{i=7}^{N-1} \frac{1}{7^{2i}+7^{i-1}} \leq
\sum_{i=7}^{N-1} \frac{1}{49^{i}}=\frac {49^{N-7}-1}{48\cdot
49^{N-1}}$. The first six elements are
$\frac{1}{50}+\frac{1}{2408}+\frac{1}{177698}+\frac1{5765144}+\frac{1}{282477650}+\frac{1}{13841304008}
\approx 0.0204239557816752$.

Letting $N$ grow to infinity both in
the numerator and denominator gives us a numerator of $6+\frac
76=\frac{43}6$ and denominator of at  most
$\frac{1}{50}+\frac{1}{2408}+\frac{1}{177698}+\frac1{5765144}+\frac{1}{282477650}+\frac{1}{13841304008}+\frac{1}{49^6
\cdot 48}+4.375 \approx 4.39542395578318$. The resulting lower
bound on the competitive ratio is $1.6304835981151$.
The same
value is obtained even if the exact sum of the series is found via
simulation. The difference is in the fourth digit after the
decimal point.

For the case without $1$-items, using $W_k=V_k$ for all $k$,
instead of
$(A_3-A_2)\cdot\OPT(J_3)+(A_2-A_1)\cdot\OPT(J_2)+A_1\cdot
\OPT(J_1)$ we have $(B_3-B_2)\cdot\OPT(J_3)+B_2\cdot
\OPT(J_{22})$, so instead of $2\cdot \frac 38+2 \cdot \frac 23+2
\cdot 1$ we have $2 \cdot \frac 38+4\cdot 1$. The denominator is
larger by $\frac 23$, and the resulting lower bound is
approximately $1.41575234447271$. The difference due to using the
exact sum of the series is in the fifth digit after the decimal
point.

\bibliographystyle{abbrv}
\bibliography{ooebp}

\end{document}